\documentclass[11pt]{article}

\usepackage[latin1]{inputenc}
\usepackage{graphicx}
\usepackage{color}
\usepackage{epsfig}
\usepackage{comment}
\usepackage{multirow}
\usepackage{amsmath}
\usepackage{amssymb}
\usepackage[hidelinks, bookmarksopen=true]{hyperref}
\usepackage{natbib}
\usepackage{booktabs}
\usepackage{longtable}
\usepackage{subfig}
\usepackage{enumitem}
\usepackage{footnote}
\usepackage{empheq}
\usepackage{amsthm}
\newtheorem{theorem}{Theorem}

\usepackage{booktabs}
\usepackage{bbm}
\usepackage{ifthen}

\usepackage{rotating}
\usepackage{float}
\usepackage{setspace}
\usepackage{fancyvrb}
\usepackage{dcolumn}
\newcommand*{\N}{\mathcal{N}}
\newcommand{\Var}{\mbox{Var}}
\newcommand{\mat}[1]{\mathrm{\mathbf{#1}}}
\newcommand*\col[1]{\mathcal{#1}}
\newcommand*{\RR}{\mathbb{R}}
\newcommand*{\kernel}{\mathcal{K}}
\renewcommand*{\vec}[1]{\boldsymbol{#1}}

\newcommand{\aalpha}{\mbox{\boldmath{$\alpha$}}}
\newcommand{\bbeta}{\mbox{\boldmath{$\beta$}}}
\newcommand{\ggamma}{\mbox{\boldmath{$\gamma$}}}
\newcommand{\ddelta}{\mbox{\boldmath{$\delta$}}}
\newcommand{\mmu}{\mbox{\boldmath{$\mu$}}}
\newcommand{\xxi}{\mbox{\boldmath{$\xi$}}}
\newcommand{\eepsi}{\mbox{\boldmath{$\epsilon$}}}
\newcommand{\SSigma}{\mbox{\boldmath{$\Sigma$}}}
\newcommand{\0}{\bf 0}
\newcommand{\1}{\bf 1}
\newcommand{\x}{\bf x}
\newcommand{\rr}{\bf r}
\newcommand{\B}{\bf B}
\newcommand{\I}{\bf I}
\newcommand{\K}{\bf K}
\newcommand{\LL}{\bf L}
\newcommand{\PP}{\bf P}
\newcommand{\Q}{\bf Q}
\newcommand{\U}{\bf U}
\newcommand{\V}{\bf V}
\newcommand{\X}{\bf X}
\newcommand{\Y}{\bf Y}
\newcommand{\W}{\bf W}
\newcommand{\Z}{\bf Z}

\renewenvironment{proof}[1][Proof.]{\begin{trivlist}
\item[\hskip \labelsep {\textit{\bfseries #1}}]}{\end{trivlist}}

\begin{document}

%***********************************************************************

\title{\textbf{Alleviating confounding in spatio-temporal areal models with an application on crimes against women in India}}

\author{Adin, A$^{1,2}$, Goicoa, T$^{1,2}$, Hodges, J.S.$^{3}$, Schnell, P.$^{4}$, and Ugarte, M.D.$^{1,2}$\\
\small { \textit{$^1$ Department of Statistics, Computer Sciences and Mathematics, Public University of Navarre, Spain.}} \\
\small {\textit{$^2$ Institute for Advanced Materials and Mathematics, InaMat$^2$, Public University of Navarre, Spain.}}\\
\small {\textit{$^3$ Division of Biostatistics, School of Public Health, University of Minnesota, Minneapolis, USA}}\\
\small {\textit{$^4$ Division of Biostatistics, College of Public Health, The Ohio State University, Columbus, USA}}\\
\small { $*$Correspondence to María Dolores Ugarte, Departamento de Estad\'istica, Inform\'atica y Matem\'aticas, } \\
\small { Universidad P\'ublica de Navarra, Campus de Arrosadia, 31006 Pamplona, Spain.} \\
\small {\textbf{E-mail}: lola@unavarra.es }}
\date{}

\makeatletter
\pdfbookmark[0]{\@title}{title}
\makeatother

\maketitle

\begin{abstract}
Assessing associations between a response of interest and a set of covariates in spatial areal models is the leitmotiv of ecological regression. However, the presence of spatially correlated random effects can mask or even bias estimates of such associations due to confounding effects if they are not carefully handled. Though potentially harmful, confounding issues have often been ignored in practice leading to wrong conclusions about the underlying associations between the response and the covariates.
In spatio-temporal areal models, the temporal dimension may emerge as a new source of confounding, and the problem may be even worse. In this work, we propose two approaches to deal with confounding of fixed effects by spatial and temporal random effects,  while obtaining good model predictions. In particular, restricted regression and an apparently -- though in fact not -- equivalent procedure using constraints are proposed within both fully Bayes and empirical Bayes approaches. The methods are compared in terms of fixed-effect estimates and model selection criteria. The techniques are used to assess the association between dowry deaths and certain socio-demographic covariates in the districts of Uttar Pradesh, India.

\end{abstract}

%Keywords: Bias; CAR prior; Dowry deaths; INLA; PQL; Space-time interactions;  Variance inflation
\bigskip

\section{Introduction}

Spatial and spatio-temporal disease mapping techniques have been widely used in epidemiology and public health. Though these analyses are somewhat descriptive, they have undoubted value as they provide information about the geographical pattern of the disease, how this pattern evolves in time, and where regions with extreme risk (high or low) are located. An overall spatio-temporal view of the phenomenon under study is extremely useful for generating hypotheses about factors that may be associated with the disease. Throughout this paper and to avoid any misleading conclusion, a risk factor is simply a predictor of one outcome (i.e. cancer, or crimes against women in this paper), which can be useful for identifying those areas with high risk requiring special actions or intervention programmes.

When covariates related to the study question are unknown, spatio-temporal areal models are a first step in developing an understanding of the disease or crime under study. Incorporating potential risk factors in a model is usually known as ecological regression, and it confers an inferential perspective on spatio-temporal areal models as it quantifies the relationship between a response and covariates (see, e.g., \citealp[chapter~5]{martinez2019disease}).

Spatio-temporal areal models are practical and valuable tools but they are not free from inconveniences. \cite{goicoa2018spatio} highlight some identifiability problems involving the intercept and the spatial and temporal main effects, and involving the main effects and the interaction. To overcome these problems, these authors propose to reparameterize the models or to use constraints, though other approaches exist \citep[see, e.g.][for identifiability issues in a class of generalized additive models]{Fahrmeir0,kneib2019test}. Another key issue in spatial and spatio-temporal areal models is  potential confounding between the fixed effects and random effects. Spatial confounding occurs when the covariates have a spatial pattern and are collinear with the spatial random effects. Although some authors warn about its effects (see, e.g., \citealp{clayton1993spatial,zadnik2006analysis}), it has been and still is often ignored in practice. Though we restrict our attention to models for areal data, the effects of spatial confounding have been studied in other areas such as causal inference (see for example \citealp{papadogeorgou2019adjusting}), or interpolation/prediction \citep{page2017estimation}. \cite{reich2006effects} show that adding a conditional autoregressive spatial random effect (CAR) to a fixed effects model can lead to a great change in the posterior mean or a great increase in the posterior variance of the fixed effects, compared to the non-spatial regression model. %The variance inflation is crucial as the fixed effect estimates may thus be extremely conservative.
More precisely, these authors reconsider the Slovenia stomach cancer data, and observe that when the CAR random effect is included in the model, the relationship between stomach cancer and the covariate is diluted. To overcome this problem, they proposed to specify the random effects as orthogonal to the fixed effects. Later, \cite{hodges2010adding} explain why such spatial confounding occurs, and show that adding spatially correlated random effects does not adjust fixed effects for spatially-structured missing covariates, as has been generally understood, though they do smooth fitted values.

Following the approach of reparameterizing random effects, \cite{schnell2019spectral} also examine the mechanisms of confounding of fixed effects by random effects but they focus on diagnostics to evaluate the effects of confounding rather than on methods to overcome it. Additional work includes \cite{hughes2013dimension},  \cite{hanks2015restricted}, and \cite{prates2019alleviating}.  Recently, \cite{khan2020restricted} observe that restricted regression offers inference similar to models without random effects. Khan and Calder, however, did not discuss model fit or prediction.

The present paper deals with spatio-temporal models with covariates and spatial, temporal, and spatio-temporal random effects, pursuing two main goals, namely, correct estimation of linear association between socio-demographic covariates and dowry death (a crime against women very specific to India), and good model predictions. Because we have data and covariates in space and time, confounding may be spatial, temporal or both, and it may mask the association of the outcome with the covariates.

This paper proposes different methods to deal with confounding but also with identifiability, as identifying spatial, temporal and spatio-temporal random effects is important for interpretation. On one hand we consider restricted spatial regression \citep{reich2006effects} and study its application to a spatio-temporal setting. On the other hand, we examine use of constraints to deal with both identifiability and confounding issues. While both methods seem to solve the confounding issue, they lead to notably different results in terms of model fit. Here, we try to disentangle why this happens.

The rest of the paper is organized as follows. Section 2 poses the spatial and the spatio-temporal models and briefly revisits identifiability and confounding issues. Section 3 proposes two procedures to alleviate confounding and identify the model: a model reparameterization followed by restricted regression, and the use of orthogonality constraints. In Section 4, both techniques are used to assess the association between dowry death and socio-demographic covariates,  and to predict spatio-temporal patterns of risk in Uttar-Pradesh, the most populated state in India.

\section{Pitfalls in spatial and spatio-temporal areal models}

Spatial and spatio-temporal models for areal data have been and still are valuable tools to give a complete picture of the status of a disease, crime or other variable of interest measured using areal counts. Although the benefit and soundness of these models are beyond any doubt, they are not free from inconveniences that should be conveniently addressed. In this paper we revisit spatial and spatio-temporal models with intrinsic conditional autoregressive (ICAR) priors for space and random walks priors for time, and focus on two issues: model identifiability and confounding of fixed effects by random effects. The first usually arises because the spatial and temporal random effects implicitly include an intercept, and the interaction term and the main effects overlap. The second arises from collinearity between fixed and random effects, which may lead to bias and variance inflation of the fixed effects and hence erroneous inference.

\subsection{Identifiability and confounding in spatial models}

This section focuses on a spatial model for areal count data that includes an intrinsic conditional autoregressive (ICAR) prior for space and highlights the identifiability and confounding issues.

Suppose the area under study (a country, a state) is divided into small areas (counties, districts) denoted by $i=1,\ldots S$, and that $O_i$ stands for the number of observed cases (death or incident cases, number of crimes) in the $i$th small area. Conditional on the relative risk $r_i$, a Poisson distribution with mean $\mu_i=e_ir_i$ is assumed for $O_i$, where $e_i$ is the number of expected cases computed using, for example, internal standardization. That is
\begin{eqnarray*}
%\label{spatial}
O_i|r_i \sim Pois(\mu_i=e_ir_i),\,\, \log\mu_i=\log e_i+\log r_i,
\end{eqnarray*}
where $\log(e_i)$ is an offset and $\log r_i$ is modeled as
\begin{equation}\label{model.1}
\log r_i=\beta_0+{\bf x}_i'{\bbeta }+\xi_i.
\end{equation}
Here $\beta_0$ is the intercept, ${\bf x}_i'=(x_{i1},\ldots,x_{ip})$ is a $p$-vector of standardized covariates in the $i$th area, ${\bbeta}=(\beta_1,\ldots,\beta_p)'$ is the $p$-vector of fixed effects coefficients, and $\xi_i$ is the spatial random effect with an ICAR prior \citep{besag1974spatial}. Then, the vector of spatial effects $\xxi=(\xi_1,\ldots,\xi_S)'$ follows the improper distribution with Gaussian kernel, $ p(\xxi) \propto \exp\left(-\frac{1}{2\sigma_\xi^2} \xxi^{'} \Q_\xi \xxi\right)$,
where ${\Q}_{\xi}$ is the $S\times S$ spatial neighbourhood matrix with $(i,j)$ element ${\Q}_{\xi (ij)}=-1$ if areas $i$ and $j$ are neighbours and 0 otherwise, and the $i^{th}$ diagonal element ${\Q}_{\xi (ii)}$ is the number of neighbours of the $i^{th}$ area.
In disease mapping studies, typically two regions are neighbours if they share a common border. Because ${\Q}_{\xi}$'s rows sum to zero, ${\Q}_{\xi}{\1}_S={\0}$ where ${\1}_S$ is a vector of ones of length $S$, so an intercept is implicit in the ICAR specification, leading to an identifiability problem with the model intercept.  We will assume the spatial map is connected, so ${\Q}_{\xi}$'s 0 eigenvalue has multiplicity 1;  more general cases are easily accommodated and omitted for simplicity. \cite{goicoa2018spatio} use this spectral decomposition of the precision matrix of the random effects to reveal the identifiability issue:
%%
%\begin{equation*}
${\Q}_{\xi} = {\U}_{\xi} {\SSigma}_{\xi} {\U}_{\xi}^{'} =
\left[{\U}_{{\xi}n} : {\U}_{{\xi}r} \right]
\left(\begin{array}{cc} 0 & {\0} \\ {\0} & \widetilde{{\SSigma}}_{\xi} \end{array} \right)
\left[\begin{array}{c} {\U}_{{\xi}n}^{'} \\ {\U}_{{\xi}r}^{'} \end{array} \right]$,
%\end{equation*}
%%
%%
where $\widetilde{\SSigma}_{\xi}$ is a diagonal matrix with the non-null eigenvalues of ${\Q}_{\xi}$ in the main diagonal, and ${\U}_{\xi}=[{\U}_{{\xi}n} : {\U}_{{\xi}r} ]$ is an orthogonal matrix with columns the eigenvectors of ${\Q}_{\xi}$. The matrix ${\U}_{\xi}$ is split into the matrix of eigenvectors having null eigenvalues, ${\U}_{\xi n}$, and the matrix of eigenvectors having non-null eigenvalues, ${\U}_{{\xi r}}$. Here the identifiability issue is clearly revealed as ${\U}_{\xi n}$ equals the vector of ones ${\1}_S$ divided by a normalizing constant. Consequently, the spatial Model \eqref{model.1} can be expressed in matrix form as
\begin{equation}\label{model.2}
\log{\rr} = {\bf 1}_S \beta_0 + {\X}{\bbeta }+{\xxi}= {\bf 1}_S \beta_0 + {\bf X}{\bbeta }+{\bf 1}_S \beta_{\xi} +{\U}_{{\xi r}}\aalpha_{\xi},
\end{equation}
where ${\rr}=(r_1,\ldots,r_S)'$, ${\bf X}=({\bf X}_1,\ldots,{\bf X}_p)$ is the $S\times p$ fixed effects design matrix (excluding the intercept) with ${\bf X}_j=(X_{1j},\ldots, X_{Sj})',\, j=1,\ldots, p$, $\beta_{\xi}={\bf 1}_S'{\xxi}$, and $\aalpha_{\xi}={\U}_{{\xi r}}'{\xxi}$, $\aalpha_{\xi}\sim \N\left({\0},\sigma^2_{\xi}\widetilde{{\SSigma}}_{\xi}^{-1}\right)$. Unlike ${\aalpha}_{\xi}$, which has a proper prior, ${\beta}_{\xi}$ has prior precision zero, leading to the identifiability issue:  two intercepts are present, the model's and the one implicit in the ICAR. Removing one redundant intercept from the linear predictor, Model \eqref{model.2} can be written as
\begin{equation*}
\log{\rr} =  {\bf 1}_S \beta_0 + {\bf X}{\bbeta }+{\U}_{{\xi r}}\aalpha_{\xi},
\end{equation*}
resolving the identifiability issue. Alternatively, a sum-to-zero constraint $\sum_{i=1}^S \xi_i=0$ can be considered.

This reparameterization does not preclude the other potential pitfall of spatial models: spatial confounding. Spatial confounding can be briefly defined as the impossibility of dissociating covariate effects from spatial random effects. As far as we know, \cite{reich2006effects} is the first paper describing how a CAR random effect can produce changes in the estimates of the fixed effects and inflate the variance compared to a non-spatial model.
Later, \cite{hodges2010adding} study the effect of spatial confounding more deeply and show that the usual belief that random effects adjust fixed-effects estimates for missing confounders cannot be sustained. They show that the variance inflation is large if the correlation is large between the covariate ${\X}_j$ and the eigenvector of the spatial matrix ${\Q}_{\xi}$ having the smallest non-null eigenvalue, that is, there is a collinearity problem.

To identify situations where confounding may be a serious issue, these authors hypothesize that the random effects will mask the association between the response and the covariate if the latter exhibits a trend in the long axis of the map. To overcome this confounding, they propose to retain in the model only the part of the random effects lying in the space orthogonal to the fixed effects; for a CAR model with normal response ${\Y}$, this is the model
\begin{equation*}
{\Y} =  {\bf 1}_S \beta_0 + {\bf X}{\bbeta }+{\LL}{\LL}'{\xxi},
\end{equation*}
or its reparameterized version
\begin{equation*}
{\Y} ={\bf 1}_S \beta_0 + {\bf X}{\bbeta }+{\LL}{\LL}'{\U}_{{\xi r}}\aalpha_{\xi},
\end{equation*}
where the columns of ${\LL}$ are eigenvectors having non-null eigenvalues (which in fact are all 1) of the projection matrix ${\I}_S-{\X}_*({\X}_*'{\X}_*)^{-1}{\X}_*'$ onto the orthogonal space of the fixed effects, ${\I}_S$ is the $S\times S$ identity matrix, and ${\X}_*=[{\bf \1}_S:{{\X}}]$. According to \cite{hodges2010adding}, this restricted spatial regression takes account of the spatial correlation without changing the estimates of the fixed effects. However, with non-normal responses, e.g., a Poisson model, this method requires adjustments. In particular, to deal with collinearity between the fixed and random effects we would use the linear predictor
\begin{equation*}
{\log\rr} =  {\bf 1}_S \beta_0 + {\bf X}{\bbeta }+{\hat{\W}}^{-1/2}{\LL}{\LL}'{\hat{\W}}^{1/2}{\xxi},
\end{equation*}
where ${\W}$ is a diagonal matrix of weights with diagonal elements $W_{ii} =\Var(O_i|\beta_0, {\bbeta},{\xxi})=\mu_i$. ${\W}$ is the weight matrix in the iteratively reweighted least squares algorithm, and to remove collinearity between the fixed and random effects we should delete the combinations of ${\hat{\W}}^{1/2}{\xxi}$ in the span of ${\hat{\W}}^{1/2}{\X}$ \citep{reich2006effects}. Accordingly, ${\LL}$ is now the matrix whose columns are the eigenvectors with non-zero eigenvalues of the orthogonal projection matrix ${\I}_s-{\hat{\W}}^{1/2}{\X}_*({\X}_*'{\hat{\W}}{\X}_*)^{-1}{\X}_*'{\hat{\W}}^{1/2}$ onto the orthogonal space of ${\hat{\W}}^{1/2}{\X}_*$. It is also possible to reparameterize the model to remove identifiability issues and specify the random effects as orthogonal to the fixed effects:
\begin{equation*}%\label{restrict.spatial}
{\log\rr} =  {\bf 1}_S \beta_0 + {\bf X}{\bbeta }+{\hat{\W}}^{-1/2}{\LL}{\LL}'{\hat{\W}}^{1/2}{\U}_{{\xi r}}\aalpha_{\xi},
\end{equation*}
Note that in practice, $\hat{\W}$ is obtained by fitting the spatial Model \eqref{model.1}.

\subsection{Identifiability and confounding in spatio-temporal models}

Now suppose that for each small area or district $i$, we have data for time periods denoted by $t=1,\ldots, T$. Similar to the spatial case, and conditional on the spatio-temporal relative risk $r_{it}$, assume the number of observed cases $O_{it}$ in area $i$ and time $t$ follows a Poisson distribution with mean $\mu_{it}=e_{it}r_{it}$ where $e_{it}$ is the number of expected cases, that is
\begin{eqnarray*}
O_{it}|r_{it} \sim Pois(\mu_{it}=e_{it}r_{it}),\,\, \log\mu_{it}=\log e_{it}+\log r_{it},
\end{eqnarray*}
where the log relative risk is now modelled as
\begin{equation}\label{model.3}
\log r_{it}=\beta_0+{\x}_{it}'{\bbeta}+\xi_i+\gamma_t+\delta_{it}.
\end{equation}
Here, ${\x}_{it}'=(x_{it1},\ldots,x_{itp})$ is a $p$-vector of standardized spatio-temporal covariates in area $i$ and time $t$, ${\bbeta}=(\beta_1,\ldots,\beta_p)'$ is the $p$-vector of fixed effect coefficients,  $\gamma_t$ is the temporal main effect, and $\delta_{it}$ is the spatio-temporal interaction term. In matrix form, Model \eqref{model.3} is
\begin{equation}\label{model.4}
\log{\rr} = {\bf 1}_{TS} \beta_0 +{\X}{\bbeta}+ ({\1}_{T}\otimes {\I}_{S}){\xxi}+ ({\I}_{T}\otimes {\1}_{S}){\ggamma}+{\I}_{TS}{\ddelta},\\
\end{equation}
where ${\rr}=(r_{11},\ldots, r_{S1},\ldots,r_{1T},\ldots,r_{ST})'$, ${\bf 1}_{TS}$ and ${\bf 1}_{T}$ are columns of ones of length $TS$ and $T$ respectively, ${\X}=({\X}_1,\ldots,{\X}_p)$ is the $TS\times p$ matrix of standardized spatio-temporal covariates with \sloppy ${\X}_j=(X_{11j},\ldots, X_{S1j},\ldots,X_{1Tj},\ldots,X_{STj})'$, $j=1,\ldots, p$, and ${\I}_{T}$ and ${\I}_{TS}$ are $T\times T$ and $TS\times TS$ identity matrices respectively.  As before, we consider an ICAR prior for the spatial random effect. For the vector of temporal random effects $\ggamma=(\gamma_1,\ldots,\gamma_T)'$, we use a first-order random walk (RW1), that is
$ p(\ggamma) \propto \exp\left(-\frac{1}{2\sigma_\gamma^2} \ggamma^{'} \Q_\gamma \ggamma\right)$,
where ${\Q}_{\gamma}$ is the RW1's structure matrix \citep[p.~95]{rue2005gaussian}. The vector of interaction random effects ${\ddelta}=(\delta_{11},\ldots,\delta_{S1},\ldots, \delta_{1T},\ldots,\delta_{ST})'$ is assumed to follow a distribution with the Gaussian kernel
$p(\ddelta) \propto \exp\left(-\frac{1}{2\sigma_\delta^2} \ddelta^{'} \Q_\delta \ddelta \right)$, and ${\Q_{\delta}}=({\Q}_{\gamma}\otimes {\Q}_{\xi})$.
This interaction term corresponds to the Type IV interaction of \cite{held2000bayesian}. We consider this type of interaction as it includes spatial and temporal dependence and consequently can produce spatial or temporal confounding. Now consider the spectral decomposition of ${\Q_{\gamma}}$:
${\Q}_{\gamma} = {\U}_{\gamma} {\SSigma}_{\gamma} {\U}_{\gamma}^{'} =
\left[{\U}_{{\gamma}n} : {\U}_{{\gamma}r} \right]
\left(\begin{array}{cc} {0} & {\0} \\
                        {\0} & \widetilde{{\SSigma}}_{\gamma}
                        \end{array} \right)
\left[\begin{array}{c} {\U}_{{\gamma}n}^{'} \\ {\U}_{{\gamma}r}^{'} \end{array} \right]$,
where ${\U}_{\gamma}=\left[{\U}_{{\gamma}n}:{\U}_{{\gamma}r} \right]$ is the $T\times T$ matrix of eigenvectors, ${\U}_{{\gamma}n}={\1}_T $  (up to a constant) is the eigenvector with null eigenvalue, ${\U}_{{\gamma}r}$ is the $T\times (T-1)$ matrix of eigenvectors with non-null eigenvalues, and $\widetilde{{\SSigma}}_{\gamma}$ is a diagonal matrix with the non-null eigenvalues in the main diagonal.
The spectral decomposition of ${\Q}_{\delta}={\Q}_{\gamma}\otimes{\Q}_{\xi}$ can be expressed as
${\Q}_{\delta}={\Q}_{\gamma}\otimes{\Q}_{\xi} = {\U}_{\delta} {\SSigma}_{\delta} {\U}_{\delta}^{'} =
\left[{\U}_{{\delta}n} : {\U}_{{\delta}r} \right]
\left(\begin{array}{cc} {0} & {\0} \\
{\0} & \widetilde{{\SSigma}}_{\delta}
\end{array} \right)
\left[\begin{array}{c} {\U}_{{\delta}n}^{'} \\ {\U}_{{\delta}r}^{'} \end{array} \right]$,
where, as in the previous cases, ${\U}_{{\delta}n}$ is the matrix of eigenvectors having null eigenvalues, ${\U}_{{\delta}r}$  is the matrix of eigenvectors having non-null eigenvalues, and $\widetilde{{\SSigma}}_{\delta}=\widetilde{{\SSigma}}_{\gamma}\otimes\widetilde{{\SSigma}}_{\xi}$ is a diagonal matrix with the non-null eigenvalues in the main diagonal. It can be shown easily that
${\U}_{{\delta}n}=[{\U}_{\gamma n} \otimes {\U}_{\xi n} : {\U}_{\gamma n} \otimes {\U}_{\xi r} : {\U}_{\gamma r} \otimes {\U}_{\xi n}],\quad {\U}_{{\delta}r} =[{\U}_{{\gamma}r}\otimes {\U}_{{\xi}r}]$.\\

\noindent
Similar to the spatial case, the spatio-temporal Model \eqref{model.4} can be expressed in matrix form as
\begin{eqnarray}\label{model.5}
\log{\rr} &=&
{\bf 1}_{TS} \beta_0 +{\X}{\bbeta}+{\1}_{TS} {\beta}_{\xi}+ ({\1}_T \otimes {\U}_{\xi r} ){\aalpha}_{\xi}+ {\1}_{TS} {\beta}_{\gamma} + ({\U}_{\gamma r} \otimes {\1}_S) {\aalpha}_{\gamma}\nonumber\\
 &+&[{\1}_{TS} : {\1}_T \otimes {\U}_{\xi r} : {\U}_{\gamma r} \otimes {\1}_S] {\bbeta}_{\delta} + ({\U}_{\gamma r} \otimes {\U}_{\xi r}) {\aalpha}_{\delta},
\end{eqnarray}
where  ${\beta}_{\gamma}={\U}_{{\gamma}n}^{'}{\ggamma}={\1}_T^{'}{\ggamma}$,  ${\aalpha}_{\gamma}={\U}_{{\gamma}r}^{'}{\ggamma}\sim\N\left({\0},\sigma^2_{\gamma}\widetilde{{\SSigma}}_{\gamma}^{-1}\right)$, ${\bbeta}_{\delta}={\U}_{{\delta}n}^{'}{\ddelta}$, and ${\aalpha}_{\delta}={\U}_{{\delta}r}^{'}{\ddelta}\sim\N\left({\0},\sigma^2_{\delta}\widetilde{{\SSigma}}_{\delta}^{-1}\right)$. The reparameterized form of Model \eqref{model.5} sheds light on the identifiability issues in spatio-temporal models as it lays bare repeated terms. Removing those superfluous terms gives the following model:
\begin{equation}\label{model.6}
\log{\rr} = {\bf 1}_{TS} \beta_0 +{\X}{\bbeta}+ ({\1}_T \otimes {\U}_{\xi r} ){\aalpha}_{\xi}+ ({\U}_{\gamma r} \otimes {\1}_S) {\aalpha}_{\gamma}+ ({\U}_{\gamma r} \otimes {\U}_{\xi r}) {\aalpha}_{\delta},
\end{equation}
which overcomes identifiability issues. For more details about this reparameterization, a generalization to RW2 priors for time, and the other interaction types described by \cite{held2000bayesian}, the reader is referred to \cite{goicoa2018spatio}, where sum-to-zero constraints are alternatively derived to achieve model identifiability.

Confounding issues are more challenging in spatio-temporal settings than in spatial settings. In a spatio-temporal model, the covariates can exhibit spatial patterns each year, or temporal patterns in each area. As the model includes both spatial and temporal random effects along with interaction terms, the source of confounding can be spatial, temporal, or both. Note that the reparameterized Model \eqref{model.6} may present confounding problems as the covariates may be collinear with the design matrix of the spatial, temporal, or spatio-temporal random effects. The following section proposes two methods to alleviate confounding in spatio-temporal models, restricted spatial regression and constraints that make the estimated random effects orthogonal to the fixed effects.

\section{Alleviating confounding in spatio-temporal models}
Constraints can be used to make the random effects orthogonal to the fixed effects and thus alleviate confounding. The idea of inducing orthogonality between the fixed and random effects is similar to restricted regression, but they have some differences that can lead to notably distinct results. This section shows that constraining the random effects to be orthogonal to the fixed effects is not equivalent to removing from the linear predictor the component of the random effects in the span of the fixed effects, as might reasonably be assumed. In some cases, these differences are hardly noticeable in the spatial case but they become important in spatio-temporal settings.

\subsection{Model reparameterization and restricted regression}

Consider Model \eqref{model.6}. This reparameterized spatio-temporal model is convenient as it solves the identifiability problems by removing repeated terms in the spatial and temporal main effects and the interaction random effects. The confounding issues are more challenging now because covariates are spatio-temporal, so we can have collinearity between the covariates and the spatial term, the covariates and the temporal term, or the covariates and the spatio-temporal interaction. Another concern is how to assess these correlations as the covariates have $T\times S$ entries while the spatial and temporal random effects have $S$ and $T$ elements. This section considers the matrix ${\LL}$ with columns that are the eigenvectors having non-null eigenvalues (all equal 1) of the projection matrix
${\PP}^c={\I}_{TS}-{\hat{\W}}^{1/2}{\X}_*({\X}_*'{\hat{\W}}{\X}_*)^{-1}{\X}_*'{\hat{\W}}^{1/2}$, where now ${\X}_*=[{\1}_{TS}:{\X}]$, and ${\X}$ is the design matrix of covariates. The matrix ${\PP}^c$ projects onto the space orthogonal to the (scaled) fixed effects $\hat{\W}^{1/2}{\X}_*$. Consider also the matrix ${\K}$, with columns that are the eigenvectors having eigenvalue 1 of the projection matrix ${\PP}={\hat{\W}}^{1/2}{\X}_*({\X}_*'{\hat{\W}}{\X}_*)^{-1}{\X}_*'{\hat{\W}}^{1/2}$. Note that ${\K}={\hat{\W}}^{1/2}{\X}_*$ and that $({\K}{\K}^{'}+{\LL}{\LL}^{'})={\I}_{TS}$;  we therefore propose the following spatio-temporal model:
\begin{eqnarray*}
 \log{\rr} &= &{\1}_{TS}\beta_0 + {\X}{\bbeta}+ {\hat\W}^{-1/2}({\K}{\K}^{'}+{\LL}{\LL}^{'}){\hat\W}^{1/2}({\1}_T \otimes {\U}_{\xi_r}){\aalpha}_{\xi} \nonumber\\
  && + {\hat\W}^{-1/2}({\K}{\K}^{'}+{\LL}{\LL}^{'}){\hat\W}^{1/2}({\U}_{\gamma_r} \otimes {\1}_S){\aalpha}_{\gamma}\nonumber\\
  &&+ {\hat\W}^{-1/2}({\K}{\K}^{'}+{\LL}{\LL}^{'}){\hat\W}^{1/2}({\U}_{\gamma_r} \otimes {\U}_{\xi_r}){\aalpha}_{\delta} .
\end{eqnarray*}

\noindent
As the terms involving the matrix ${\K}$ are in the span of the fixed effects ${\X}_*$, they are removed and the model becomes
\begin{eqnarray}\label{model.7}
 \log{\rr} &= &{\1}_{TS}\beta_0 + {\X}{\bbeta} +{\hat\W}^{-1/2}{\LL}{\LL}^{'}{\hat\W}^{1/2}({\1}_T \otimes {\U}_{\xi_r}){\aalpha}_{\xi} \nonumber\\
  &&+ {\hat\W}^{-1/2}{\LL}{\LL}^{'}{\hat\W}^{1/2}({\U}_{\gamma_r} \otimes {\1}_S){\aalpha}_{\gamma}\nonumber\\
  && + {\hat\W}^{-1/2}{\LL}{\LL}^{'}{\hat\W}^{1/2}({\U}_{\gamma_r} \otimes {\U}_{\xi_r}){\aalpha}_{\delta}.
\end{eqnarray}

\noindent
Model \eqref{model.7} deserves comment. First, by using the matrix ${\LL}$, {\it all} random effects have been restricted to be orthogonal to the fixed effects (${\LL}$ is orthogonal to ${\hat\W}^{1/2}{\X}_*$ by construction, that is ${\X}_*^{'}{\hat\W}^{1/2}{\LL}={\bf 0}$), but we could also restrict only some of the random effects. Restricting all the random effects may not be necessary if only some of the random effects confound a fixed effect. Thus it is possible to orthogonalize only the spatial, or temporal, or interaction random effects.  The other important issue is the final estimation of the spatial and temporal random effects. Compared to Model \eqref{model.6}, the spatial and temporal random effects in Model \eqref{model.7} undergo a substantial change. In Model \eqref{model.6}, the spatial and temporal main effects are $({\1}_T \otimes {\U}_{\xi r} ){\aalpha}_{\xi}$ and $({\U}_{\gamma r} \otimes {\1}_S) {\aalpha}_{\gamma}$ respectively;  clearly $S$ spatial effects are repeated in time, while $T$ temporal effects are repeated for the $S$ small areas. In the orthogonalized Model \eqref{model.7}, however, the spatial and temporal main effects are ${\hat\W}^{-1/2}{\LL}{\LL}^{'}{\hat\W}^{1/2}({\1}_T \otimes {\U}_{\xi_r}){\aalpha}_{\xi}$ and ${\hat\W}^{-1/2}{\LL}{\LL}^{'}{\hat\W}^{1/2}({\U}_{\gamma_r} \otimes {\1}_S){\aalpha}_{\gamma}$ respectively. Because these terms include the matrix ${\LL}$, which has $T\times S$ rows (the covariates are spatio-temporal, in general), the spatial effect associated with the $i^{th}$ area is different in each time period as it depends on the value of the covariates in that period. Similarly, the effect of time period $t$ is also different for each area as it depends on the covariates in that area. Consequently, Model \eqref{model.7} has time-varying spatial effects and space-varying temporal effects.

\subsection{Constraints to alleviate confounding}
\label{constraints.alleviate}

Placing constraints is a way to achieve identifiability in spatio-temporal disease mapping models \citep{goicoa2018spatio}. Constraints can be used not just to identify models but also to alleviate spatial confounding, by constraining the random effects to be orthogonal to the fixed effects. This section considers such constraints. Specifically, consider Model \eqref{model.4} and consider the linear predictor $\hat{\W}^{1/2}{\1}_{TS}\beta_0+\hat{\W}^{1/2}{\X}{\bbeta}+{\hat\W}^{1/2}({\1}_{T}\otimes {\I}_{S}){\xxi}+{\hat\W}^{1/2}({\I}_{T}\otimes {\1}_{S}){\ggamma}+{\hat\W}^{1/2}{\I}_{TS}{\ddelta}$. To make the random effects orthogonal to the fixed effects, these constraints are required:
\begin{eqnarray}\label{ortho.cons}
\left[{\1}_{TS}:{\X}\right]^{'}{\hat\W}({\1}_T \otimes {\I}_S){\xxi}={\0} & \Longleftrightarrow & \sum_{i=1}^S \xi_i \hat{w}_{i\cdot} =0, \quad  \sum_{i=1}^S \xi_i(\hat{w}x_j)_{i\cdot}=0, \,\forall j=1,\ldots, p,\nonumber\\
\left[{\1}_{TS}:{\X}\right]^{'}{\hat\W}({\I}_T \otimes {\1}_S){\ggamma}={\0}     & \Longleftrightarrow & \sum_{t=1}^T \gamma_t \hat{w}_{\cdot t} =0 ,\quad  \sum\limits_{t=1}^T\gamma_t(\hat{w}x_j)_{\cdot t}=0, \,\forall j=1,\ldots, p,\nonumber\\
\left[{\1}_{TS}:{\X}\right]^{'}{\hat\W}{\ddelta}={\0} & \Longleftrightarrow &  \sum_{i=1}^S \sum_{t=1}^T\hat{w}_{it}\delta_{it}=0,\quad \sum_{i=1}^S \sum\limits_{t=1}^T\hat{w}_{it}x_{it}\delta_{it}=0, \,\forall j=1,\ldots, p,\nonumber\\
\end{eqnarray}
where $\hat{w}_{i\cdot}=\sum_{t=1}^T\hat{w}_{it}$, $\hat{w}_{\cdot t}=\sum_{i=1}^S\hat{w}_{it}$, $(\hat{w}x_j)_{i\cdot}=\sum\limits_{t=1}^T\hat{w}_{it}x_{jit}$, and $(\hat{w}x_j)_{\cdot t}=\sum_{i=1}^S\hat{w}_{it}x_{jit}$. That is, the spatial random effect is constrained to be orthogonal to the time-averaged covariates at each location, the temporal random effect is constrained to be orthogonal to the space-averaged covariates at each time, and the interaction random effect is constrained to be orthogonal to the full fixed-effects design matrix. However, the interaction term ${\ddelta}$ is confounded with the spatial and temporal main (random) effects and these additional constraints are required:
\begin{equation}\label{ortho.adicional}
  [({\1}_T \otimes {\I}_S) : ({\I}_T \otimes {\1}_S)]^{'}{\hat\W}{\ddelta}={\0} \Longleftrightarrow \begin{array}{l}\sum\limits_{i=1}^S \hat{w}_{it}\delta_{it}=0 \quad \forall t=1,\ldots,T. \\[2.ex] \sum\limits_{t=1}^T \hat{w}_{it}\delta_{it}=0 \quad \forall i=1,\ldots,S. \end{array},
\end{equation}
making the constraint $\sum_{i=1}^S \sum_{t=1}^T\hat{w}_{it}\delta_{it}=0$ redundant. Avoiding confounding between the interaction and the spatial and temporal main effect terms is crucial for model interpretation. The spatial main effects capture spatial variation that is not accounted for by the covariates so they may help identify spatial risk factors that have not been included in the model. Similarly, the temporal main effects may help identify risk factors associated with time. The interaction random effects usually explain a small portion of variability and capture deviations from the main effects. Given that the interaction term's null space equals the design matrix of the spatial and temporal main effects (see Model \eqref{model.5}), confounding between main effect and interaction random effects may make it impossible to achieve one of the goals of disease mapping: smoothing in space and time to uncover spatial and temporal patterns.

Though the restricted regression and the constraints approaches would seem to be equivalent, they are in fact different. The restricted regression approach focuses on the reduced Model \eqref{model.7}, where the part of the random effects in the span of the fixed effects has been removed, and orthogonality is achieved through the matrix ${\LL}$. The constraints approach, by contrast, starts with the full Model \eqref{model.4} and Equation \eqref{ortho.cons}'s constraints force the random effects to be orthogonal to the fixed effects, which is a way to remove collinearities between them. A key distinction between the restricted regression and constraints approaches is that in the former, the spatial and temporal effects change in time and space respectively due to the spatio-temporal nature of the matrix ${\LL}$, while in the constraints approach, the spatial effects remain constant in time and the temporal effects do not change in space.

One may think the methods are equivalent because removing the part of the random effects in the span of the fixed effects means ${\K}^{'}{\hat\W}^{1/2}({\1}_T \otimes {\I}_{S}){\xxi}={\0}$, ${\K}^{'}{\hat\W}^{1/2}({\I}_{T} \otimes {\1}_S){\ggamma}={\0}$, and ${\K}^{'}{\hat\W}^{1/2}{\I}_{TS}{\ddelta}={\0}$, which match the constraints in Equation \eqref{ortho.cons}, taking into account that ${\K}={\hat\W}^{1/2}[{\1}:{\X}]$. However, placing constraints is in fact equivalent to {\it oblique projections} of the random effects. That is, the constraints in Equations \eqref{ortho.cons} and \eqref{ortho.adicional} correspond to this model:
\begin{eqnarray*}\label{model.proj}
\log{\rr} &=& {\1}_{TS}\beta_0 + {\X}\beta + ({\1}_T \otimes {\I}_S){\PP}_{\xi}{\xxi} + ({\I}_T \otimes {\1}_S){\PP}_{\gamma}{\ggamma} + {\PP}_{\delta}{\ddelta}\nonumber\\
         &=& {\1}_{TS}\beta_0 + {\X}\beta + ({\1}_T \otimes {\PP}_{\xi}){\xxi} + ({\PP}_{\gamma} \otimes {\1}_S){\ggamma} + {\PP}_{\delta}{\ddelta}
\end{eqnarray*}
where
\begin{equation}\label{projections}
{\PP}_{\xi} = {\LL}_{\xi}[{\LL}_{\xi}^{'}{\Q}_{\xi}{\LL}_{\xi}]^{-1}{\LL}_{\xi}^{'}{\Q}_{\xi},\quad
{\PP}_{\gamma} = {\LL}_{\gamma}[{\LL}_{\gamma}^{'}{\Q}_{\gamma}{\LL}_{\gamma}]^{-1}{\LL}_{\gamma}^{'}{\Q}_{\gamma}, \quad
{\PP}_{\delta} = {\LL}_{\delta}[{\LL}_{\delta}^{'}{\Q}_{\delta}{\LL}_{\delta}]^{-1}{\LL}_{\delta}^{'}{\Q}_{\delta}
\end{equation}
are made to be, respectively, oblique projections onto the orthogonal complements of the row spaces of the matrices
\begin{equation}\label{matrixB}
{\B}_{\xi} = {\X}_*^{'}{\hat\W}({\1}_T \otimes {\I}_S),\quad {\B}_{\gamma} ={\X}_*^{'}{\hat\W}({\I}_T \otimes {\1}_S),\quad {\B}_{\delta} = [({\1}_T \otimes {\I}_S) : ({\I}_T \otimes {\1}_S):{\X}]^{'}{\hat\W}
\end{equation}
by setting ${\LL}_{\xi}$ to be the matrix whose columns are the eigenvectors with non-zero eigenvalues of the projection matrix ${\I}_{S} - {\B}_{\xi}^{'}({\B}_{\xi} {\B}_{\xi}^{'})^{-1}{\B}_{\xi}$, and similarly for ${\LL}_{\gamma}$ and ${\LL}_{\delta}$. See Supplementary Material A for the equivalence of the oblique projections and the constraints and detailed expressions of matrices ${\B}_{\xi}$, ${\B}_{\gamma}$, and ${\B}_{\delta}$.

\subsection{Model fitting and inference}

Two main methods have been used to fit spatial and spatio-temporal disease mapping models: a fully Bayesian approach and an empirical Bayes approach. The latter provides point estimates of quantities of interest, traditionally using penalized quasi-likelihood (PQL; \citealp{breslow1993approximate}). It has proven to be interesting because it is relatively simple and has few convergence problems, and it has been used to fit different models such as the ICAR or P-splines \citep{dean2001detecting,dean2004penalized,ugarte2010spatio, ugarte2012splines}. However, PQL automatically places sum-to-zero constraints due to the rank deficiency of the random effects covariance matrices, and placing additional constraints is not so straightforward (see the Appendix for full details).

The fully Bayesian approach is probably the most-used technique for model fitting and inference because it provides a full posterior distribution for quantities of interest. Though it has traditionally relied on Markov chain Monte Carlo (MCMC), the computational burden of MCMC prompted development of other attractive procedures, including INLA (integrated nested Laplace approximations;  \citealp{rue2009approximate}). INLA's main advantage is that it provides approximate Bayesian inference without using MCMC, leading to substantial reduction in computational cost. INLA is ready to use in the free software \texttt{R} using the package \texttt{R-INLA}, which has implemented general models that can be adapted to disease mapping. Also, imposing constraints in INLA is relatively simple.

\section{Data analysis: the association between dowry deaths and socio-demographic covariates in Uttar Pradesh, India}

In this section, the two approaches to alleviate confounding are used to assess the potential association between dowry deaths and some socio-demographic covariates in Uttar Pradesh, the most populated state in India.
Dowry deaths is a cruel form of violence against women deep-rooted in India. It is strongly related to dowry, which can be defined as the amount of money, property, or goods that the bride's family gives to the groom or his relatives for the marriage. Although dowry was originally designed to protect women from unfair traditions, such as the impossibility of women owning immovable property (see \citealp{banerjee2014dowry}), it has become a means by which the husband or husband's relatives extort higher dowries under the threat of physical violence against the wife. This form of violence can be extended over time, ending in what is known as a dowry death. If a woman commits suicide because she has experienced mental or physical violence related to the dowry, this is also considered a dowry death.

We have data on dowry deaths in 70 districts of Uttar Pradesh, the Indian state with the highest rate of dowry deaths, during the period 2001-2014. In 2014, the last year of the study period, 8,455 dowry deaths were registered representing 29.2\% of all dowry deaths in India. One of the difficulties of combatting this crime against women is the lack of knowledge about potential risk factors that might be associated with dowry deaths, and thus help to predict this crime. One hypothesized risk factor is the sex ratio, that is, the number of females per 1000 males. The literature has contradictory results about the sex ratio:  some authors find a negative association between dowry deaths and sex ratio \citep{mukherjee2001crimes}, while others find a positive association \citep{dang2018dowry}.  In a more in-depth study of dowry deaths in Uttar Pradesh, \cite{vicente2020crime} consider spatio-temporal models and include some potential risk factors as covariates to assess their association with dowry deaths. However, these authors do not address confounding issues.  We now focus on some of those covariates, namely sex ratio ($x_1)$, population density ($x_2)$, female literacy rate ($x_3)$, per capita income ($x_4)$, murder rate ($x_5)$, and burglary rate ($x_6)$, estimating their association with dowry deaths accounting for confounding. %As a note of caution, sex ratio, population density, and female literacy rate are population-based measures and are only available in the census years 2001 and 2011. For other years, they have been linearly interpolated.
The goal is to see the effect of confounding on the estimates and standard errors of the fixed effects. We also compare restricted regression and constraints in terms of model fit and complexity using DIC (Deviance Information Criterion, \citealp{spiegelhalter2002bayesian}) and AIC (Akaike Information Criterion, \citealp{akaike1974new}) in a Bayesian and frequentist approach respectively. %, WAIC (Watanabe 2010), or LS (Logarithmic Score, Gneiting and Raftery, 2007).

\subsection{Data analysis}
\label{sec:Data_Analysis}

For purely spatial models, \cite{reich2006effects} and \cite{hodges2010adding} argued that spatial confounding is created by a high correlation between a covariate and the eigenvector of the spatial precision matrix having the smallest non-null eigenvalue. In the spatio-temporal setting, this suggests examining analogous correlations for the spatial, temporal, and spatio-temporal random effects.  Therefore, before fitting any models, we examine the data and compute some correlations. Specifically, we split the covariates into spatial vectors for each year of the period, and computed Pearson's correlation between those spatial vectors and ${\U}_{\xi_{69}}$, the eigenvector of the spatial precision matrix with the smallest non-null eigenvalue, so for each covariate, we compute fourteen spatial correlations.  Similarly, we split the covariates into temporal vectors for each district, and evaluate the correlations between those temporal vectors and ${\U}_{\gamma_{13}}$, the eigenvector of the temporal precision matrix with the smallest non-null eigenvalue, seventy temporal correlations for each covariate. Boxplots of correlations between the covariates and the spatial eigenvector ${\U}_{\xi_{69}}$ for each year (left) and correlations between the covariates and the temporal eigenvector ${\U}_{\gamma_{13}}$ for each area (right) are displayed in Figure \ref{figB1} in Supplementary Material B, together with some explanation.\\

\noindent
Now we consider four spatio-temporal models.
\begin{itemize}
  \item Model ST1: Simple spatio-temporal Poisson model (Equation \eqref{model.4} without random effects)
  \item Model ST2: Spatio-temporal model without accounting for confounding (Equations \eqref{model.4} and \eqref{model.6} for INLA and PQL respectively).
  \item Model ST3: Spatio-temporal model with restricted regression (Equation \eqref{model.7}). %or equivalently Model \eqref{model.8})
  \item Model ST4: Spatio-temporal model with orthogonality constraints  (Equations \eqref{ortho.cons}, \eqref{ortho.adicional})
\end{itemize}

\noindent
Within the full Bayesian approach, we have used Normal prior distributions with mean 0 and variance equal to 1000 for fixed effects. Regarding random effects, we have used uniform prior distributions on the positive real line for the standard deviations. We have also considered logGamma(1,0.00005) priors for the log-precisions and PC priors for the standard deviations \citep[see e.g.][]{simpson2017}, but these gave similar results and hence they are not shown here to save space.

\begin{table}[!t]
\caption{Posterior mean, posterior standard deviation, and 95\% credible intervals of the fixed effects for models fitted with INLA (left), and point estimates, standard errors  and 95\% confidence intervals obtained with PQL (right). The results for Model ST2 are in bold when the covariate is not significantly associated with dowry death. \label{tab1}}
\begin{center}
\footnotesize
\resizebox{\textwidth}{!}{
\begin{tabular}{ll|rrrr|rrrr}
\\[-4ex]
& & \multicolumn{4}{c|}{INLA (simplified Laplace)} & \multicolumn{4}{c}{PQL (tol=1e-5)}\\[0.5ex]
\hline
&&\multicolumn{8}{c}{Sex ratio}\\[0.5ex]\hline
&Model & Mean & SD & $q_{0.025}$ & $q_{0.975}$ & Estimate & SE & $q_{0.025}$ & $q_{0.975}$ \\[0.5ex]
\hline
$\beta_1$
& ST1  & -0.2366 &0.0085 &   -0.2532 &   -0.2200 & -0.2366 &0.0085 &-0.2532 &-0.2200  \\
& \bf ST2  & \bf-0.0920 &\bf0.0470 &  \bf -0.1836 &   \bf 0.0015 &  -0.0924 & 0.0459 & -0.1825 & -0.0024 \\
& ST3  & -0.2303 &0.0085 &   -0.2470 &   -0.2135 & -0.2299 &0.0082 &-0.2459 &-0.2139 \\
& ST4  & -0.2284 &0.0082 &   -0.2445 &   -0.2125 & -0.2281 &0.0081 &-0.2441 &-0.2121 \\
\hline
&&\multicolumn{8}{c}{Population Density}\\[0.5ex]\hline
$\beta_2$
& ST1  & -0.0917 &0.0065 &   -0.1044 &   -0.0791 & -0.0917 &0.0065 &-0.1044 &-0.0792   \\
& \bf ST2  & \bf-0.0069 &\bf0.0318 &   \bf-0.0682 &    \bf0.0568 & \bf-0.0081 &\bf0.0304 &\bf-0.0677 & \bf0.0515  \\
& ST3  & -0.0901 &0.0066 &   -0.1031 &   -0.0773 & -0.0904 &0.0061 &-0.1024 &-0.0785  \\
& ST4  & -0.0946 &0.0063 &   -0.1071 &   -0.0824 & -0.0949 &0.0063 &-0.1073 &-0.0826  \\
\hline
&&\multicolumn{8}{c}{Female literacy rate}\\[0.5ex]\hline
$\beta_3$
& ST1  &  0.0992 &0.0076 &    0.0843 &    0.1141 &  0.0992 &0.0076 & 0.0843 &0.1141   \\
& \bf ST2  & \bf-0.0478 &\bf0.0501 &   \bf-0.1482 &    \bf0.0487 & \bf-0.0469 &\bf0.0474 &\bf-0.1398 &\bf0.0460  \\
& ST3  &  0.0946 &0.0080 &    0.0789 &    0.1104 &  0.0946 &0.0077 & 0.0796 &0.1096  \\
& ST4  &  0.0975 &0.0077 &    0.0823 &    0.1126 &  0.0975 &0.0077 & 0.0823 &0.1126  \\
\hline
&&\multicolumn{8}{c}{Per capita income}\\[0.5ex]\hline
$\beta_4$
& ST1  & -0.0661 &0.0084 &   -0.0827 &   -0.0498 & -0.0661 &0.0084 &-0.0827 &-0.0499  \\
&\bf ST2  &\bf -0.0196 &\bf0.0296 &   \bf-0.0776 &    \bf0.0385 & \bf-0.0198 &\bf0.0288 &\bf-0.0763 & \bf0.0368 \\
& ST3  & -0.0651 &0.0085 &   -0.0819 &   -0.0486 & -0.0650 &0.0081 &-0.0809 &-0.0491 \\
& ST4  & -0.0680 &0.0080 &   -0.0837 &   -0.0525 & -0.0678 &0.0079 &-0.0833 &-0.0522 \\
\hline
&&\multicolumn{8}{c}{Murder rate}\\[0.5ex]\hline
$\beta_5$
& ST1  & 0.0833 &0.0076 &    0.0682 &    0.0982 & 0.0833 &0.0076 &0.0682 &0.0982  \\
& ST2  & 0.0846 &0.0203 &    0.0446 &    0.1244 & 0.0845 &0.0202 &0.0449 &0.1241 \\
& ST3  & 0.0906 &0.0081 &    0.0748 &    0.1064 & 0.0907 &0.0077 &0.0756 &0.1057 \\
& ST4  & 0.0881 &0.0079 &    0.0726 &    0.1034 & 0.0881 &0.0079 &0.0727 &0.1035 \\
 \hline
 &&\multicolumn{8}{c}{Burglary rate}\\[0.5ex]\hline
$\beta_6$
& ST1  & 0.0419 &0.0062 &    0.0297 &    0.0541 & 0.0419 &0.0062 &0.0297 &0.0540 \\
& ST2  & 0.0535 &0.0161 &    0.0220 &    0.0850 & 0.0535 &0.0158 &0.0226 &0.0844 \\
& ST3  & 0.0424 &0.0068 &    0.0291 &    0.0557 & 0.0423 &0.0063 &0.0300 &0.0547 \\
& ST4  & 0.0431 &0.0063 &    0.0307 &    0.0554 & 0.0429 &0.0063 &0.0306 &0.0552 \\
\hline
\end{tabular}}
\end{center}
\end{table}

Table \ref{tab1} shows the posterior mean and standard deviation of the fixed effects with a 95\% credible interval, computed using INLA (simplified Laplace strategy). Point estimates obtained with PQL are also displayed with their standard error and 95\% confidence interval. For sex ratio, the estimate with Model ST2 and INLA is about 40\% (in absolute value) of the estimates obtained with Models ST1, ST3, and ST4, and the posterior SD is 5.5 times higher. For Model ST2, the 95\% credible interval includes 0, while the intervals from the other models are far from zero. The results with PQL are similar although Model ST2's confidence interval for sex ratio barely excludes 0. For population density, the estimate from Model ST2 is less than 10\% of the estimates obtained with Models ST1, ST3, and ST4, with INLA or PQL, and again the posterior SD with Model ST2 is about 5 times higher than with the other models. The consequence is that using Model ST2, the association between population density and dowry deaths is not significant. Per capita income has similar results. The effect of confounding on the estimated association with female literacy rate is also noteworthy:  with Model ST2, the estimated effect is negative (though not significant) whereas with the rest of models is positive and significant. The estimated associations with murder rate and burglary rate are similar for the four models, though the posterior SD is clearly larger for Model ST2. These results are revealing and illustrate the potential harmful consequences of ignoring the effects of confounding: the estimated association between the response and the covariate may be diluted or dramatically changed. Posterior estimates of the hyperparameters (standard deviations) are displayed in Table \ref{tab1B} in Supplementary Material B. Similar estimates for the standard deviations are obtained with Models ST2, ST3, and ST4.

\begin{table}[!t]
\caption{Model selection criteria and computing time (seconds) for spatio-temporal models fit with INLA and PQL. Computations were made on a twin superserver with four processors, Intel Xeon 6C and 96GB RAM, using the R-INLA (stable) version 19.09.03. \label{tab2}}
\begin{center}
\resizebox{\textwidth}{!}{
\begin{tabular}{l|rrrc|rrrc}
\\[-4ex]
& \multicolumn{4}{c|}{INLA (simplified Laplace)} & \multicolumn{4}{c}{PQL (tol=1e-5)}\\[0.5ex]
\hline
& $\bar{D}$ & $p_D$ & DIC & Time & Deviance & $Df$ & AIC & Time \\[0.5ex]
\hline
Model ST1  & 8471.72 &   7.18 & 8478.90 &   3 & 8466.20 &   7.00 & 8480.20 &   1 \\
Model ST2  & 5962.60 & 239.26 & 6201.86 &  18 & 5727.89 & 236.48 & 6200.85 &  74 \\
Model ST3  & 5962.59 & 239.27 & 6201.86 & 192 & 5727.85 & 236.46 & 6200.77 &  86 \\
Model ST4  & 6492.87 & 232.01 & 6724.88 &  22 & 6269.53 & 227.64 & 6724.81 & 180 \\
\hline
\end{tabular}}
\end{center}
\end{table}

Spatio-temporal models accounting for confounding (Models ST3 and ST4) lead to similar estimates of the fixed effects and their posterior standard deviations but they differ in terms of model selection criteria. Table \ref{tab2} displays the mean deviance ($\bar{D}$), the effective number of parameters $p_D$, DIC and computing time for the INLA fits. For the PQL fit, the deviance, the number of parameters ($DF$) and AIC are provided. As expected, the model without random effects has the worst fit;  the six covariates are not enough to explain the data's variation. For Models ST3 and ST4, the difference in fit is remarkable. Clearly, Model ST3 provides a much better fit:  the differences in $\bar{D}$ and DIC are about 500 points.
The INLA fit with a simplified Laplace strategy and using constraints (Model ST4) is much faster than restricted regression (22 and 192 seconds respectively). Computing time for Model ST3 in INLA has been reduced (about one half at least) by plugging in the posterior modes of the hyperparameters obtained from Model ST2 as initial values. Note that the total time required to fit Model ST3 should include the computing time of Model ST2 (18 seconds in our data analysis).
Similarly, computing times for Models ST3 and ST4 in PQL have been reduced about 23\% using the variance components estimates obtained in Model ST2 as initial values in the estimation algorithm. In this case, the computing times shown in Table \ref{tab2} correspond to the total time needed to fit the corresponding models. Additionally, Supplementary Material B provides scatter plots of the estimated relative risks from Models ST2, ST3, and ST4 fitted with INLA and PQL (Figure \ref{figB2});  posterior spatial patterns (top row), posterior temporal patterns (middle row) obtained from Models ST2, ST3, and ST4 fitted with INLA, and posterior spatio-temporal patterns (see \citealp{adin2017smoothing}) for three districts, Agra, Balrampur, and Gautam Buddha Nagar (bottom row) in Figure \ref{figB3}. Finally, INLA relative risk estimates (posterior means) obtained with models ST3 and ST4 are shown in Figure \ref{figB4} for the same three districts.

Why do Models ST3 and ST4 fit so differently? First, the spatial and temporal terms in Models ST4 are $({\1}_T \otimes {\I}_S ){\xxi}$ and $({\I}_T \otimes {\1}_S) {\ggamma}$ respectively. Consequently the spatial effects are repeated in every year and the temporal effects are repeated in each area. In Model ST3, however, the spatial and temporal terms $({\1}_T \otimes {\U}_{\xi r} ){\xxi}$ and $({\U}_{\gamma r} \otimes {\1}_S){\ggamma}$ are premultiplied by ${\hat\W}^{-1/2}{\LL}{\LL}^{'}{\hat\W}^{1/2}$. Because the matrix ${\LL}$ contains spatio-temporal information --- it depends on the spatio-temporal covariate ${\X}$ --- the spatial effect for the $i$th area is different in each time period and the temporal effect in year $t$ is also different for each area. Second, no substantial change in relation to Model ST2 is made in Model ST3. Basically, the random effects in this latter model are split into two pieces, one in the span of the covariates and one orthogonal to the covariates. When Model ST3 removes the part in the span of the covariates, it is simply discarding redundant information. By contrast, Model ST4 changes the model by forcing the random effects to be orthogonal to the fixed effects. Moreover, the spatial random effects lie in the space orthogonal to the time-weighted-added covariates, and the temporal random effects are orthogonal to the spatial-weighted-added covariates. This is equivalent to an oblique projection onto the orthogonal subspace of the fixed effects, unlike Model ST3, where the projection is orthogonal. Because the orthogonal projection minimizes the distance between the original random effects and the projection, this could explain the improvements in fit over the oblique projection (constraints).

Finally, note that we also fit Model ST3 without premultiplying the temporal and spatio-temporal effects by ${\hat\W}^{-1/2}{\LL}{\LL}^{'}{\hat\W}^{1/2}$ and the results are nearly identical, indicating that all confounding arises from the spatial term. We also fit Model ST3 premultiplying by ${\hat\W}^{-1/2}{\LL}{\LL}^{'}{\hat\W}^{1/2}$ only to the temporal and the spatio-temporal effects, and the confounding effects are not avoided.

\section{Discussion}

Including spatially correlated random effects in a model can seriously affect inference about fixed effects due to confounding. This is particularly dramatic in ecological spatial regression where the main objective is to estimate associations between the response variable and certain covariates. These relationships can be masked due to bias and variance inflation of the fixed effects caused by confounding.  Though documented in the literature, spatial confounding has generally been ignored in applications and this practice has carried over to spatio-temporal settings. Here we study confounding in spatio-temporal ecological models in which including temporally correlated random effects and space-time interaction random effects (spatially and temporally correlated) can exacerbate confounding problems.  We have considered two procedures to remedy the potentially harmful effects of confounding, restricted regression and constraints.

In light of this paper's results, we would like to emphasize some points and provide some guidelines to practitioners. First, the relative risk estimates are not affected by confounding, so if the relative risks are of primary interest, ignoring confounding is not a problem. Second, both restricted regression and orthogonal constraints alleviate confounding and provide rather similar estimates of the fixed effects and their standard errors. However, the two approaches differ importantly in terms of model selection criteria and computing time. Although the constraints approach is computationally more efficient than restricted regression (in INLA), the latter gives clearly better fits. Consequently, if the target of the analysis is to establish associations between risk factors and the phenomenon under study, along with studying spatio-temporal patterns of risk, we recommend using restricted regression. We have also observed that differences between the approaches become more evident as the number of covariates increases. The difference in fits may arise because the deviance, and hence DIC (and AIC), are directly related to the orthogonal projection, the projection used in restricted regression. Because the orthogonal projection minimizes least squares, and $\bar{D}$ (mean deviance) is a transform of least squares, any deviation from the orthogonal projection gives worse mean deviance and worse DIC as long as the effective number of parameters $p_D$ does not change much. Consequently, the more the oblique projection (constrained approach) deviates from the orthogonal projection, the worse the mean deviance and thus DIC. In the data analysis considered here, the constrained approach has the smallest value of $p_D$ because some degrees of freedom are removed due to the constraints, but DIC is worse because the reduction in $p_D$ does not compensate for the increase in mean deviance.

A desirable property of the constrained approach is that it keeps the (random) spatial and temporal main effects constant in time and space respectively while restricting all random effects to be orthogonal to the fixed effects. However, at least in the data set analysed here, this behavior comes at an important price in terms of model fit, unlike other proposals such as, e.g., \cite{hughes2013dimension}. It is not easy to guess when the constrained and restricted regression approaches will lead to similar fits. We suspect that if the covariates do not have a substantial spatio-temporal interaction the constrained approached could work well. This is consistent with the observation that in a spatial analysis of our data for the year 2011 (not shown), both procedures are nearly equivalent. Moreover, fitting spatio-temporal models including only female literacy rate (a covariate with scarcely any spatio-temporal interaction), the difference in fit between both approaches is reduced considerably.

Both procedures have been fitted using a fully Bayesian and a classical approach. Though INLA provides the posterior distributions of all quantities and hence the maximum information, PQL is still a valuable tool that allows model fitting in a reasonable time providing essential information to understand the phenomenon under study. Computing times for restricted regression can be reduced in INLA by plugging in as initial values the posterior modes of the hyperparameters obtained from the spatio-temporal model with confounding. Similarly, restricted regression can be sped up in PQL by using as initial values the variance parameter estimates obtained from the spatio-temporal model with confounding.

To sum up, models with spatial, temporal, and spatio-temporal random effects lead to a good fit of the spatio-temporal patterns of risk, but they fail to account for the correct linear association between the response and the covariates. Models including only covariates provide unbiased estimates of the fixed effects coefficients but they give a poor fit as they do not capture the variability unexplained by the covariates. The restricted regression considered in this paper offers in one single model the best of the two approaches: similar to the model without random effects, it provides correct estimates of the fixed effects but substantially improves model fit and prediction of spatio-temporal patterns of risk in line with the model with random effects.  This is an advantage of the restricted regression approach over the models without random effects and classical spatio-temporal models. Recent research \citep{khan2020restricted} concludes that for Gaussian models credible intervals for fixed effects obtained with restricted regression are nested within credible intervals obtained with models without spatial random effects. These authors acknowledge that restricted regression inference on the coefficients is very similar to the one from non spatial models. In the Poisson spatio-temporal case study analysed here, credible intervals for some of the coefficients in the restricted regression approach are not nested within those of the model without random effects, but differences are small. Further research is needed to investigate this issue more in depth. The use of constraints also provides correct estimates of the fixed effects and improves model fit, but to a lesser extent than the restricted regression approach, at least in the data set analysed here. We would like to highlight that in this work, associations between the covariates and the response should be understood as correlations, with the final objective of identifying good predictors of the response, and not as causal relationships within a causal-inference framework. This latter approach is beyond the scope of this paper.

Finally, we want to emphasize the consequences of ignoring confounding for the data analysed here. Dowry death in India, particularly in Uttar Pradesh, is a complex problem for which risk factors (socio-demographic, economic, cultural or religious) are not yet clearly identified. Ignoring confounding may lead researchers to discard some potential risk factors and wrongly estimate their associations with dowry death. In this paper, sex ratio, population density, female literacy rate, per capita income, murder rate, and burglary rate were found to be associated with dowry deaths when confounding was taken into account. Ignoring such effects masks the association between dowry deaths and some of those risk factors, which obscures understanding of this atrocious practice that takes the lives of thousands of women in India.

%%% Acknowledgements (if any)
%%% ------------------------------------------
\section*{Acknowledgements}
This work has been supported by Project MTM2017-82553-R (AEI/FEDER, UE). It has also been
partially funded by \lq\lq la Caixa'' Foundation (ID 1000010434), Caja Navarra Foundation, and UNED Pamplona, under agreement LCF/PR/PR15/51100007.

%%%%%%%%%%%%%%%%%%%%%%%%%%%%%%%%%%%%%%%%%%%%%%%%%%%%%%%%%%%%%%%%%%%%%%%%%%%%%%%%%%%%%%%%%%%%%%%%%%
%%%%%%%%%%%%%%%%%%%%%%%%%%%%%%%%%%%%%%%%%%%%%%%%%%%%%%%%%%%%%%%%%%%%%%%%%%%%%%%%%%%%%%%%%%%%%%%%%%
%%%%%%%%%%%%%%%%%%%%%%%%%%%%%%%%%%%%%%%%%%%%%%%%%%%%%%%%%%%%%%%%%%%%%%%%%%%%%%%%%%%%%%%%%%%%%%%%%%
%%%%%%%%%%%%%%%%%%%%%%%%%%%%%%%%%%%%%%%%%%%%%%%%%%%%%%%%%%%%%%%%%%%%%%%%%%%%%%%%%%%%%%%%%%%%%%%%%%

%**************% %  SOFTWARE  % %**************%
\section*{Software}
\label{Software}

Software in the form of R code, together with a sample input data set and complete documentation will be available at \url{https://github.com/spatialstatisticsupna/Confounding_article}.

%%%%%%%%%%%%%%%%%%%%%%%%%%%%%%%%%%%%%%%%%%%%%%%%%%%%%%%%%%%%%%%%%%%%%%%%%%%%%%%%%%%%%%%%%%%%%%%%%%

%**************% %  Appendix  % %**************%

\section*{Appendix}
This subsection is about adding orthogonality constraints in the PQL approach, in particular about modifying the algorithm to replace the sum-to-zero identifiability constraints with new constraints that identify the model and also achieve the desired orthogonality. Here we use ``conditioning by kriging'' \citep[pp.~37 and 93]{rue2005gaussian} where the covariance matrix of the random effects conditional on the constraints is used. Consider the covariance matrices

{\small
\begin{eqnarray*}
\text{Cov}({\xxi}|{\B}_{\xi}{\xxi}=0) & = & {\Q}_{\xi}^{-}-{\Q}_{\xi}^{-}{\B}_{\xi}^{'}({\B}_{\xi}{\Q}_{\xi}^{-}{\B}_{\xi}^{'})^{-1}{\B}_{\xi}{\Q}_{\xi}^{-}\\[1.ex]
\text{Cov}({\ggamma}|{\B}_{\gamma}{\ggamma}=0) & = & {\Q}_{\gamma}^{-}-{\Q}_{\gamma}^{-}{\B}_{\gamma}^{'}({\B}_{\gamma}{\Q}_{\gamma}^{-}{\B}_{\gamma}^{'})^{-1}{\B}_{\gamma}{\Q}_{\gamma}^{-} \\[1.ex]
\text{Cov}({\ddelta}|{\B}_{\delta}{\ddelta}=0) & = & {\Q}_{\delta}^{-}-{\Q}_{\delta}^{-}{\B}_{\delta}^{'}({\B}_{\delta}{\Q}_{\delta}^{-}{\B}_{\delta}^{'})^{-1}{\B}_{\delta}{\Q}_{\delta}^{-}
\end{eqnarray*}
}

\noindent where ``$^-$" denotes the Moore-Penrose generalized inverse, and the matrices ${\B}_{\xi}$, ${\B}_{\gamma}$, and ${\B}_{\delta}$ capture the orthogonality constraints of the previous section and are defined in Equation \eqref{matrixB}.
Using these covariance matrices, the PQL algorithm would automatically place the desired orthogonality constraints. However, due to the rank deficiency of the matrices ${\Q}_{\xi}$, ${\Q}_{\gamma}$, and ${\Q}_{\delta}$, the usual sum-to-zero constraints are also imposed in addition to the weighted sum-to-zero constraints. To overcome that problem, consider instead these covariance matrices (see Theorem 1 in Supplementary Material A for details):

{\small
\begin{eqnarray*}%\label{covariance.changed}
\text{Cov}({\xxi}|{\B}_{\xi}{\xxi}={\0}) & = & {\V}_{\xi} = {\LL}_{\xi}[{\LL}_{\xi}^{'}{\Q}_{\xi}{\LL}_{\xi}]^{-1}{\LL}_{\xi}^{'}, \nonumber\\%[1.ex]
\text{Cov}({\ggamma}|{\B}_{\gamma}{\ggamma}={\0}) & = & {\V}_{\gamma} = {\LL}_{\gamma}[{\LL}_{\gamma}^{'}{\Q}_{\gamma}{\LL}_{\gamma}]^{-1}{\LL}_{\gamma}^{'}, \nonumber \\%[1.ex]
\text{Cov}({\ddelta}|{\B}_{\delta}{\ddelta}={\0}) & = & {\V}_{\delta} = {\LL}_{\delta}[{\LL}_{\delta}^{'}{\Q}_{\delta}{\LL}_{\delta}]^{-1}{\LL}_{\delta}^{'}, \nonumber
\end{eqnarray*}
}

\noindent where ${\LL}_{\xi}$, ${\LL}_{\gamma}$, and ${\LL}_{\delta}$ are as in Equation \eqref{projections}.
%\bigskip
Doing this, the null spaces of the covariance matrices are now spanned by the vectors of constraints, that is the rows of the matrices ${\B}_{\xi}$, ${\B}_{\gamma}$, and ${\B}_{\delta}$,  and the PQL algorithm automatically circumvents identifiability issues and provides estimates satisfying the orthogonality requirements. To see this briefly, the PQL algorithm requires a working vector
\begin{eqnarray*}%\label{PQL.working}
 {\boldsymbol O}^*={\X}_*{\bbeta}+{\Z}_{\xi}{\xxi}+{\Z}_{\gamma}{\ggamma}+{\Z}_{\delta}{\ddelta}+({\boldsymbol O}-{\mmu})g'({\mmu}),
\end{eqnarray*}
where  ${\Z}_{\xi}={\1}_{T}\otimes {\I}_{S}$, ${\Z}_{\gamma}={\I}_{T}\otimes {\1}_{S}$, and ${\Z}_{\delta}={\I}_{TS}$ are the design matrices of the spatial and temporal main effects and the interaction effect respectively; $g'(\mu)=1/\mu$ is the derivative of the link function $g$, which here is the logarithmic function, ${\eepsi}=({\boldsymbol O}-{\mmu})g'({\mmu})\sim N({\0},{\hat\W}^{-1})$, and ${{\hat\W}}=diag(\mu_{it})$. Then the fixed effect estimator (including the intercept) is $\hat{\bbeta}=({\X}_*^{'}\hat{\V}^{-1}{\X}_*)^{-1}{\X}_*^{'}\hat{\V}^{-1}{\boldsymbol O}^*$, where ${\V}={\W}^{-1}+{\Z}_{\xi}{\V}_{\xi}{\Z}_{\xi}^{'}+{\Z}_{\gamma}{\V}_{\gamma}{\Z}_{\gamma}^{'}+{\Z}_{\delta}{\V}_{\delta}{\Z}_{\delta}^{'}$
and the random effects are estimated as
$$  \hat{\xxi}= \hat{\V}_{\xi}{\Z}_{\xi}^{'}\hat{\V}^{-1}({\boldsymbol O}^*-{\X}_*\hat{\bbeta}),\,  \hat{\ggamma} = \hat{\V}_{\gamma}{\Z}_{\gamma}^{'}\hat{\V}^{-1}({\boldsymbol O}^*-{\X}_*\hat{\bbeta}),\,
 \hat{\ddelta} = \hat{\V}_{\delta}{\Z}_{\delta}^{'}\hat{\V}^{-1}({\boldsymbol O}^*-{\X}_*\hat{\bbeta}).$$

\noindent
Clearly, ${\B}_{\xi}{\hat{\xxi}}=\mat{0}$,  ${\B}_{\gamma}\hat{\ggamma}=\mat{0}$, and ${\B}_{\delta}\hat{\ddelta}=\mat{0}$, as the rows of ${\B}_{\xi}$, ${\B}_{\gamma}$, and ${\B}_{\delta}$, span the null spaces of ${\V}_{\xi}$, ${\V}_{\gamma}$, and ${\V}_{\delta}$ respectively.

%%% References if bibTeX is used
%%%
%%% Please, do not specify any \bibliographystyle{} command!
%%%
%%% It is already specified in the smj.cls and its
%%% second specification here causes error.
%%% ------------------------------------------------------------

%\bibliography{biblio_confounding}

%%% References (if created by hand).

\newpage

%%%%%%%%%%%%%%%%%%%%%%%%%%%%%%%%%%%%%%%%%%%%%%%%%%%%%%%%%%%%%%%%%%%%%%%%%%%
% Supplementary material
%%%%%%%%%%%%%%%%%%%%%%%%%%%%%%%%%%%%%%%%%%%%%%%%%%%%%%%%%%%%%%%%%%%%%%%%%%%
\setcounter{section}{0} %\setcounter{equation}{0}
\renewcommand{\thesection}{\Alph{section}}

\setcounter{figure}{0}
\renewcommand\thefigure{\thesection.\arabic{figure}}

\setcounter{table}{0}
\renewcommand\thetable{\thesection.\arabic{table}}

%**************% %  Supplementary material % %**************%
\section{Supplementary Material A}
\label{appendixA}

This section shows that placing constraints is equivalent to an oblique projection. In particular, we focus on the situation studied in the paper where constraints on Gaussian variables (the random effects) are specified by a precision matrix without constraints on the null space. In the paper, the precision matrices of the random effects are rank deficient and usually sum-to-zero constraints corresponding to the null space are required to fit the model. However, as we are making the random effects orthogonal to the fixed effects (intercept included), the usual sum-to-zero constraints are not required and must be replaced with weighted sum-to-zero constraints.

\begin{theorem}\label{Theo1}
  Let $\vec{Y}$ be a random variable of length $n$ with density
  \begin{equation}
    p_Y(\vec{y}) \propto \exp\left[-\frac{1}{2} \vec{y}^{'} \mat{Q} \vec{y}\right],
  \end{equation}
  with $\mat{Q}$ not necessarily of full rank.
  Let $\col{A}$ be a subspace of ${\RR}^n$ such that $\col{A} \cap \kernel(\mat{Q}) = \{\vec{0}\}$, and $\kernel(\mat{Q})$ stands for the null space (kernel) of ${\Q}$.
  Let $\mat{A}$ and $\mat{B}$ be matrices with rows that form orthonormal bases for $\col{A}$ and its orthogonal complement $\col{A}^\perp$, respectively.
  Finally, let
  \begin{equation}
    \mat{P}_A = \mat{A}^{'} (\mat{A} \mat{Q} \mat{A}^{'})^{-1} \mat{A} \mat{Q}.
  \end{equation}
  Then the following distributions are equal:
  \begin{enumerate}
  \item $[\vec{Y} | \mat{B} \vec{Y} = \vec{0}]$;
  \item $[\vec{Y} | \vec{Y} \in \col{A}]$;
  \item $[\mat{P}_A \vec{Y}]$;
  \item $\N[\vec{0}, \mat{A}^{'} (\mat{A} \mat{Q} \mat{A}^{'})^{-1} \mat{A}]$.
  \end{enumerate}
\end{theorem}

\begin{proof}
  Let $\mat{M} = (\mat{A}^{'} \; \mat{B}^{'})^{'}$ be orthogonal where $\mat{A}$ and $\mat{B}$ are $(n-c)\times n$ and $c\times n$ matrices, and define $\vec{z} = \mat{M}\vec{y}$.
  Then
  \begin{equation}
    \begin{aligned}
      p_Z(\vec{z})
      &= p_Y(\mat{M}^{-1} \vec{z}) |\mat{M}^{-1}|, \\
      &\propto \exp\left[-\frac{1}{2} \vec{z}^{'} (\mat{M}^{-1})^{'}\mat{Q}\mat{M}^{-1} \vec{z} \right], \\
      &= \exp\left[-\frac{1}{2} \vec{z}^{'} (\mat{M}\mat{Q}\mat{M}^{'}) \vec{z} \right].
    \end{aligned}
  \end{equation}
  Letting
  \begin{equation*}
    \begin{pmatrix}
        \mat{K} & \mat{L}
      \end{pmatrix}
      \begin{pmatrix}
        \mat{R} & \mat{0} \\
        \mat{0} & \mat{0}
      \end{pmatrix}
      \begin{pmatrix}
        \mat{K}^{'} \\ \mat{L}^{'}
      \end{pmatrix}
  \end{equation*}
  be the spectral decomposition of $\mat{Q}$, where $\mat{K}$ and $\mat{L}$ are matrices with eigenvectors having non-null and null eigenvalues respectively, and $\mat{R}$ is a diagonal matrix with the non-null and positive eigenvalues, note that
  \begin{equation*}
    \begin{aligned}
      \mat{M} \mat{Q} \mat{M}^{'}
      &=
      \begin{pmatrix}
        \mat{A}\mat{Q}\mat{A}^{'} & \mat{A}\mat{Q}\mat{B}^{'} \\
        \mat{B}\mat{Q}\mat{A}^{'} & \mat{B}\mat{Q}\mat{B}^{'}
      \end{pmatrix} \\
      &=
      \begin{pmatrix}
        \mat{A} \\ \mat{B}
      \end{pmatrix}
      \begin{pmatrix}
        \mat{K} & \mat{L}
      \end{pmatrix}
      \begin{pmatrix}
        \mat{R} & \mat{0} \\
        \mat{0} & \mat{0}
      \end{pmatrix}
      \begin{pmatrix}
        \mat{K}^{'} \\ \mat{L}^{'}
      \end{pmatrix}
      \begin{pmatrix}
        \mat{A}^{'} & \mat{B}^{'}
      \end{pmatrix}.
    \end{aligned}
  \end{equation*}
  Then, $\mat{A}\mat{Q}\mat{A}^{'} = \mat{A} \mat{K} \mat{R} \mat{K}^{'} \mat{A}^{'}$.
  Let $\vec{x} \neq \vec{0}$.
  Since $\mat{A}$ is of full row rank, $\mat{A}^{'} \vec{x} \neq \vec{0}$, and $\mat{K}^{'} \mat{A}^{'} \vec{x} \neq \vec{0}$ as long as $\mat{A}^{'}\vec{x} \not\in \col{R}(\mat{K}^{'})^\perp = \col{R}(\mat{L}^{'}) = \kernel(\mat{Q})$, where $\col{R}()$ indicate the row space of a matrix.
  Since $\mat{A}^{'}\vec{x} \in \col{R}(\mat{A})$, and $\col{R}(\mat{A}) \cap \kernel(\mat{Q}) = \{\vec{0}\}$, $\mat{K}^{'} \mat{A}^{'} \vec{x} \neq \vec{0}$.
  Since $\mat{R}$ is trivially positive definite, $\vec{x}^{'} \mat{A}\mat{K}\mat{R}\mat{K}^{'}\mat{A}^{'}\vec{x} > 0$, thus $\mat{A}\mat{Q}\mat{A}^{'} = \mat{A}\mat{K}\mat{R}\mat{K}^{'}\mat{A}^{'}$ is positive definite and therefore invertible.
  We have
  \begin{equation}
    \begin{aligned}
      p(\vec{z}_{[1, n-c]} | \mat{B} \vec{y} = \vec{0})
      &= p(\vec{z}_{[1, n-c]} | \vec{z}_{[n-c+1, n]} = \vec{0}), \\
      &\propto  \exp\left[-\frac{1}{2} \vec{z}_{[1, n-c]}^{'} \mat{A}\mat{Q}\mat{A}^{'} \vec{z}_{[1, n-c]} \right],
    \end{aligned}
  \end{equation}
  where $\vec{z}_{[a, b]}=(z_a,\ldots,z_b)^{'}$. Then, since $\mat{A}\mat{Q}\mat{A}^{'}$ is invertible,
  \begin{equation}
      [\vec{Y} | \mat{B} \vec{Y} = \vec{0}]
      \sim \N\left[\vec{0}, \mat{V} \right], \quad
      \mat{V}
      =
      \begin{pmatrix}
        \mat{A}^{'} & \mat{B}^{'}
      \end{pmatrix}
      \begin{pmatrix}
        (\mat{A}\mat{Q}\mat{A}^{'})^{-1} & \vec{0} \\
        \vec{0}^{'} & \mat{0}
      \end{pmatrix}
      \begin{pmatrix}
        \mat{A} \\ \mat{B}
      \end{pmatrix}
      =
      \mat{A}^{'} (\mat{A}\mat{Q}\mat{A}^{'})^{-1}\mat{A}.
    \end{equation}

\noindent
When the Gaussian random variable $\vec{Y}$ is specified via the precision matrix $\mat{Q}$, $\kernel(\mat{Q})$ contains the unidentified degrees of freedom in the sense that $f_Y(\vec{y} + \vec{b}) = f_Y(\vec{y})$ for all $\vec{b} \in \kernel(\mat{Q})$.
  Note that $\kernel(\mat{P}_A) = \kernel(\mat{Q})$, so $\mat{P}_A (\vec{y} + \vec{b}) = \mat{P}_A \vec{y}$ for all $\vec{b} \in \kernel(\mat{Q})$, so $\mat{P}_A \vec{Y}$ is identified because changing $\vec{Y}$ adding $\vec{b}$ does not change $\mat{P}_A\vec{Y}$.
  Thus, when computing $\Var[\mat{P}_A \vec{Y}]$, we may restrict $\vec{Y}$ to $\kernel(\mat{Q})^\perp$ so that $\Var[\vec{Y}] = \mat{Q}^-$.
  Then,
  \begin{equation*}
    \begin{aligned}
      \Var[\mat{P}_A \vec{Y}]
      &=
      \left[\mat{A}^{'} (\mat{A}\mat{Q}\mat{A}^{'})^{-1}\mat{A}\mat{Q}\right] \mat{Q}^- \left[\mat{A}^{'} (\mat{A}\mat{Q}\mat{A}^{'})^{-1}\mat{A}\mat{Q}\right]^{'}, \\
      &=
      \mat{A}^{'} (\mat{A}\mat{Q}\mat{A}^{'})^{-1} \mat{A} \mat{Q}\mat{Q}^-\mat{Q}\mat{A}^{'}(\mat{A}\mat{Q}\mat{A}^{'})^{-1}\mat{A}, \\
      &=
      \mat{A}^{'} (\mat{A}\mat{Q}\mat{A}^{'})^{-1}\mat{A}\mat{Q}\mat{A}^{'}(\mat{A}\mat{Q}\mat{A}^{'})^{-1}\mat{A}, \\
      &= \mat{A}^{'}(\mat{A}\mat{Q}\mat{A}^{'})^{-1}\mat{A}.
    \end{aligned}
  \end{equation*}
\end{proof}

\noindent
Notes:\\

\noindent
In the situation described in the paper, we need $\mat{A}_{\xi}$ and $\mat{B}_{\xi}$, $\mat{A}_{\gamma}$ and $\mat{B}_{\gamma}$, and $\mat{A}_{\delta}$ and $\mat{B}_{\delta}$.
The $\mat{B}$ matrices (the constraints matrices) are given by Equation \eqref{matrixB}, and each $\mat{A}$ matrix may be constructed by taking its rows to be the eigenvectors of the orthogonal projection matrix $\mat{I}_n - \mat{B}^{'} (\mat{B} \mat{B}^{'})^{-1} \mat{B}$ whose eigenvalues are 1 (equivalently, non-zero).
Since the projection is orthogonal, its matrix is symmetric and admits a set of orthonormal eigenvectors forming a basis for ${\RR}^n$, and its null space $\col{R}(\mat{B})$ is orthogonal to its row space, the eigenvectors whose eigenvalues are non-zero.
The condition $\col{R}(\mat{A}) \cap \kernel(\mat{Q}) = \{\vec{0}\}$ ensures that the constraint $\mat{B}\vec{Y} = \vec{0}$ is sufficient to identify $\vec{Y}$.

\underline{Expressions for constraint matrices ${\B}_{\xi}$, ${\B}_{\gamma}$, and ${\B}_{\delta}$ }
{\footnotesize
\begin{eqnarray*}%\label{matrixB}
  {\B}_{\xi} &=& {\X}_*^{'}{\hat\W}({\1}_T \otimes {\I}_S)=\left(
                                                           \begin{array}{ccc}
                                                           {\hat w}_{1.} & \ldots & {\hat w}_{S.} \\
                                                             ({\hat w}x_1)_{1.} & \cdots & ({\hat w}x_1)_{S.} \\
                                                             \vdots & \ddots & \vdots \\
                                                             ({\hat w}x_p)_{1.} & \cdots & ({\hat w}x_p)_{S.} \\
                                                           \end{array}
                                                         \right),\nonumber\\
  {\B}_{\gamma} &=&{\X}_*^{'}{\hat\W}({\I}_T \otimes {\1}_S)=\left(
                                                           \begin{array}{ccc}
                                                            {\hat w}_{.1} & \ldots & {\hat w}_{.T} \\
                                                             ({\hat w}x_1)_{.1} & \cdots & ({\hat w}x_1)_{.T} \\
                                                             \vdots & \ddots & \vdots \\
                                                             ({\hat w}x_p)_{.1} & \cdots & ({\hat w}x_p)_{.T} \\
                                                           \end{array}
                                                         \right), \nonumber\\
  {\B}_{\delta} &=& [({\1}_T \otimes {\I}_S) : ({\I}_T \otimes {\1}_S):{\X}]^{'}{\hat\W}=\left(
                                                           \begin{array}{ccccccc}
                                                             {\hat w}_{11}&\cdots& 0&\cdots&{\hat w}_{1T}&\cdots &0 \\
                                                            \vdots& \ddots &\vdots&\cdots&  \vdots&  \ddots&\vdots \\
                                                             0&\cdots& {\hat w}_{S1}& \cdots & 0&\cdots& {\hat w}_{ST}\\
                                                            {\hat w}_{11}&\cdots& {\hat w}_{S1}&\cdots&0&\cdots &0 \\
                                                            \vdots& \ddots &\vdots&\cdots&  \vdots&  \ddots&\vdots \\
                                                             0&\cdots& 0& \cdots & {\hat w}_{1T}&\cdots& {\hat w}_{ST}\\
                                                            x_{111}{\hat w}_{11}&\cdots &x_{1S1}{\hat w}_{S1}&\cdots&x_{11T}{\hat w}_{1T}&\cdots& x_{1ST}{\hat w}_{ST} \\
                                                             \vdots& \ddots &\vdots&\cdots&  \vdots&  \ddots&\vdots \\
                                                            x_{p11}{\hat w}_{11}&\cdots &x_{pS1}{\hat w}_{S1}&\cdots&x_{p1T}{\hat w}_{1T}&\cdots &x_{pST}{\hat w}_{ST} \\
                                                           \end{array}
                                                         \right)\nonumber\\
\end{eqnarray*}
}

%\newpage

%**************% %  APPENDIX B  % %**************%
\section{Supplementary Material B}
\label{appendixB}

\begin{figure}[!h]
\vspace{-0.5cm}
\includegraphics[width=0.9\textwidth]{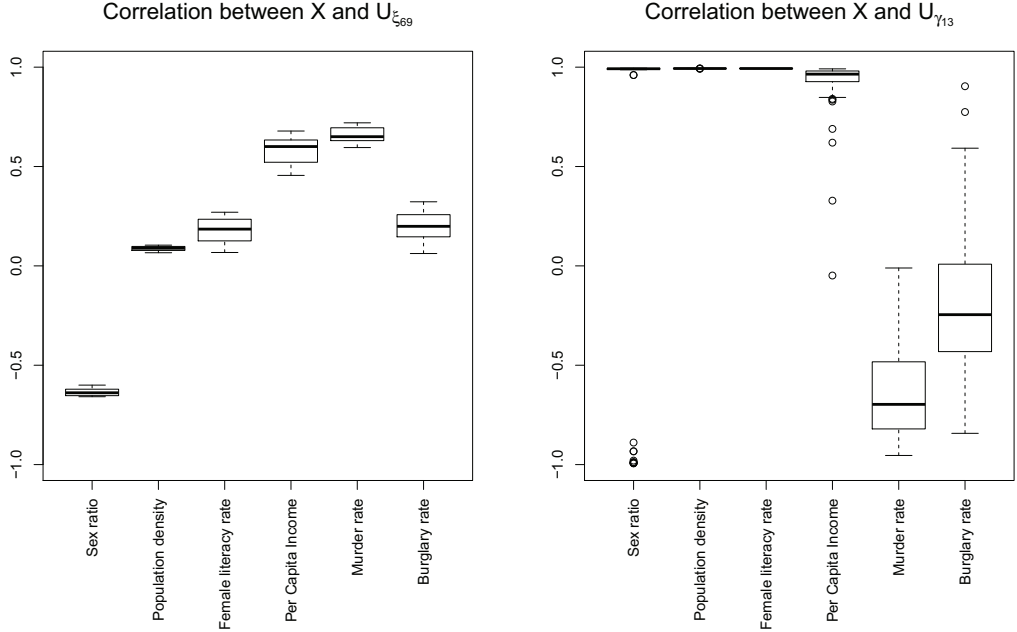}
\caption{Boxplots of correlations between the covariates and the spatial eigenvector ${\U}_{\xi_{69}}$ for each year (left) and correlations between the covariates and the temporal eigenvector ${\U}_{\gamma_{13}}$ for each area (right).}
\label{figB1}
\end{figure}

Figure \ref{figB1} displays boxplots of correlations between the covariates and the spatial eigenvector ${\U}_{\xi69}$ for each year (left picture), and boxplots of correlations between the covariates and the temporal eigenvector ${\U}_{\gamma13}$ for each area (right picture). Sex ratio, per capita income and murder rate exhibit the highest spatial correlations while the other covariates show moderate or low correlations. Regarding temporal correlations, the population-based variables (sex ratio, population density, and female literacy rate) and per capita income show the highest correlations. Population-based covariates exhibit temporal correlations close to 1 or $-1$ because they are only available at census years and have been linearly interpolated for the other years, and the temporal eigenvector ${\U}_{\gamma13}$ is nearly a straight line.

%%%%%%%%%%%%%%%%%%%% Hyper parameters

\begin{table}[!t]\footnotesize
\caption{Posterior estimates of the hyperparameters obtained using uniform priors on the positive real line for the standard deviations (INLA), and point estimates obtained with PQL}\label{tab1B}
\begin{center}
%\footnotesize
%\renewcommand\arraystretch{0.6}
%\resizebox{\textwidth}{!}{
\begin{tabular}{ll|rrrr|rrrr}
& & \multicolumn{4}{c|}{INLA (simplified Laplace)} & \multicolumn{4}{c}{PQL (tol=1e-5)}\\[0.5ex]
\hline
&Model & Mean & SD & $q_{0.025}$ & $q_{0.975}$ & Estimate & SE & $q_{0.025}$ & $q_{0.975}$ \\[0.5ex]
\hline$\sigma^2_s$
%& ST1  & $-$    &$-$    &    $-$    &    $-$    & $-$    &$-$    &$-$    &$-$    \\
& ST2  & 0.2072 & 0.0420&    0.1391 &    0.3021 & 0.1961 & 0.0370& 0.1236& 0.2685\\
& ST3  & 0.2072 & 0.0420&    0.1391 &    0.3021 & 0.1961 & 0.0370& 0.1237& 0.2686\\
& ST4  & 0.2349 & 0.0466&    0.1595 &    0.3403 & 0.2218 & 0.0422& 0.1391& 0.3044\\
\hline
$\sigma^2_t$
%& ST1  & $-$    &$-$    &    $-$    &    $-$    & $-$    &$-$    &$-$    &$-$    \\
& ST2  & 0.0174 & 0.0086&     0.0065&     0.0390& 0.0131 & 0.0056& 0.0022& 0.0241\\
& ST3  & 0.0174 & 0.0086&     0.0065&     0.0390& 0.0131 & 0.0056& 0.0022& 0.0241\\
& ST4  & 0.0164 & 0.0135&     0.0036&     0.0519& 0.0093 & 0.0057& 0.0000& 0.0204\\
\hline
$\sigma^2_{st}$
%& ST1  & $-$    &$-$    &    $-$    &    $-$    & $-$    &$-$    &$-$    &$-$    \\
& ST2  & 0.0212 & 0.0038&     0.0147&     0.0294& 0.0207 & 0.0035& 0.0139& 0.0275\\
& ST3  & 0.0212 & 0.0038&     0.0147&     0.0294& 0.0207 & 0.0035& 0.0139& 0.0275\\
& ST4  & 0.0222 & 0.0039&     0.0155&     0.0306& 0.0213 & 0.0036& 0.0142& 0.0283\\
\hline
\end{tabular}%}
\end{center}
\end{table}

Table \ref{tab1B} displays posterior estimates of the standard deviations of the random effects obtained with INLA, and point estimates of the standard deviations obtained with PQL. The estimates are rather similar in all models and in general, INLA estimates do not differ much from those obtained with PQL. Here uniform priors on the real line have been used for the standard deviations, but similar results were obtained with logGamma(1,0.00005) priors on the log-precisions.

%%%%%%%%%%%%%%%%%%%%

\begin{figure}[!t]
\includegraphics[width=\textwidth]{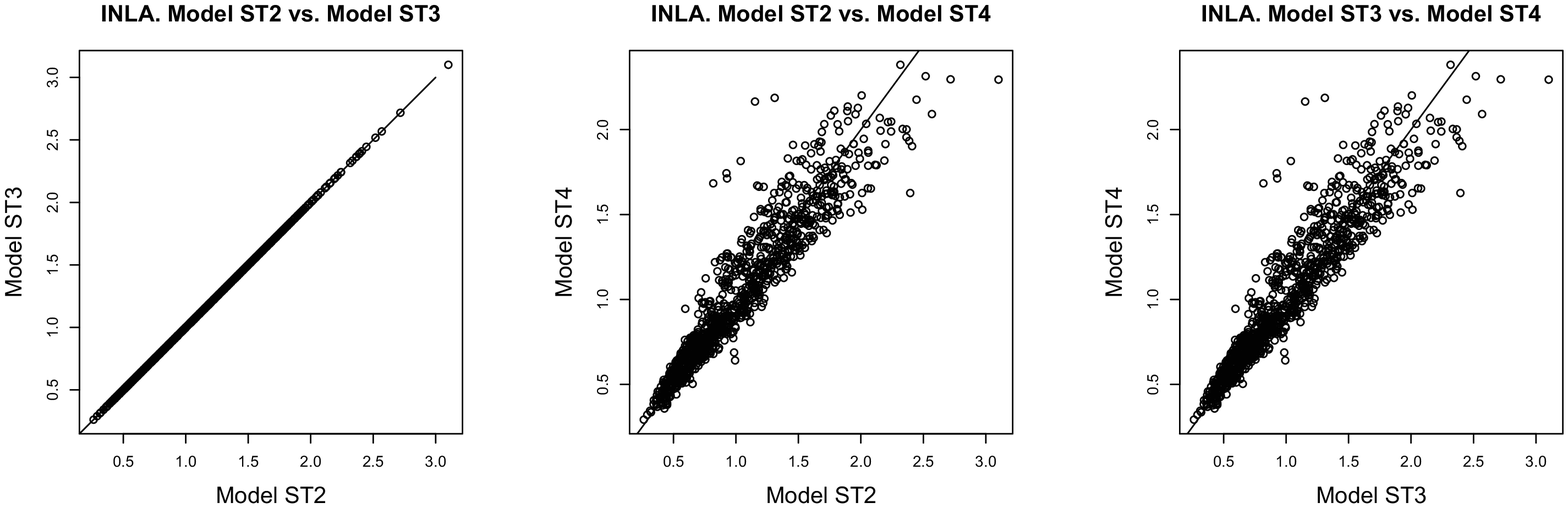}\vfill
\includegraphics[width=\textwidth]{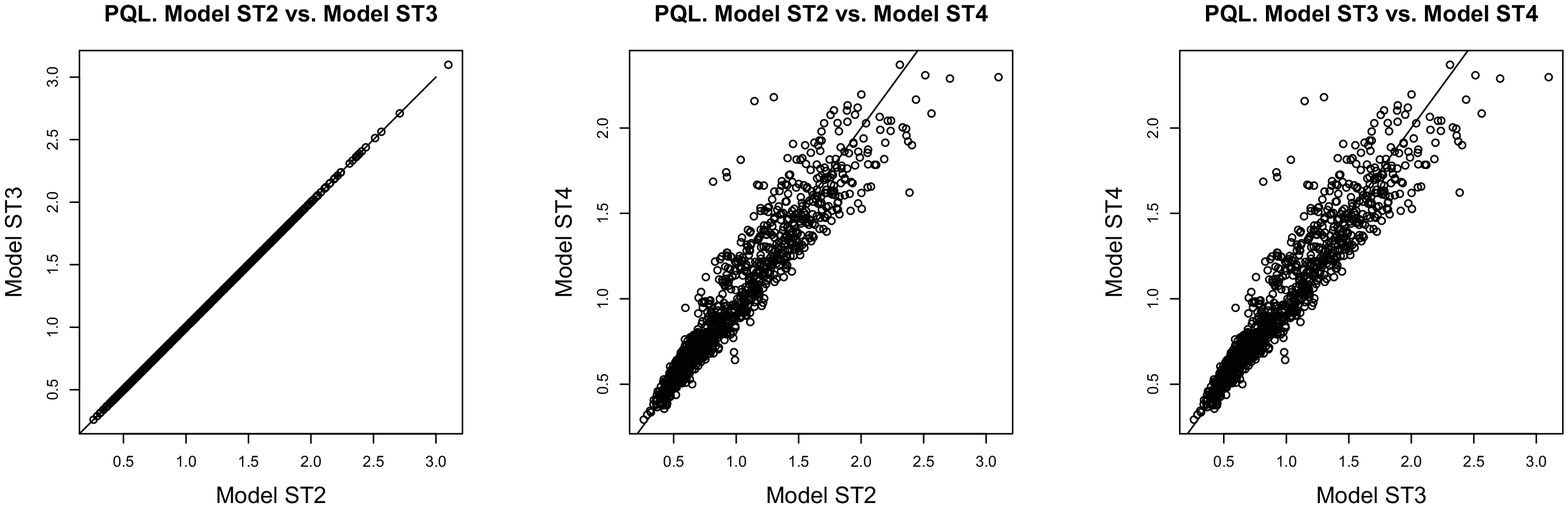}
\caption{Scatter plots of relative risk estimates obtained from Models ST2, ST3, and ST4. Top row: posterior means estimated with INLA;  bottom row:  point estimates estimated with PQL.}
\label{figB2}
\end{figure}

Figure \ref{figB2}  shows scatter plots of the estimated relative risks from Models ST2, ST3, and ST4 fitted with INLA (posterior means, top row) and PQL (point estimates, bottom row).
For both fitting techniques, Model ST2 (spatio-temporal model with no correction for confounding) shows the same fit as Model ST3 (accounting for confounding using restricted regression). However, comparing Models ST2 and ST3 with Model ST4 (accounting for confounding using constraints) shows notable differences:  the two methods that deal with confounding give different fits.

%\newpage

\begin{figure}[!t]
\includegraphics[width=\textwidth]{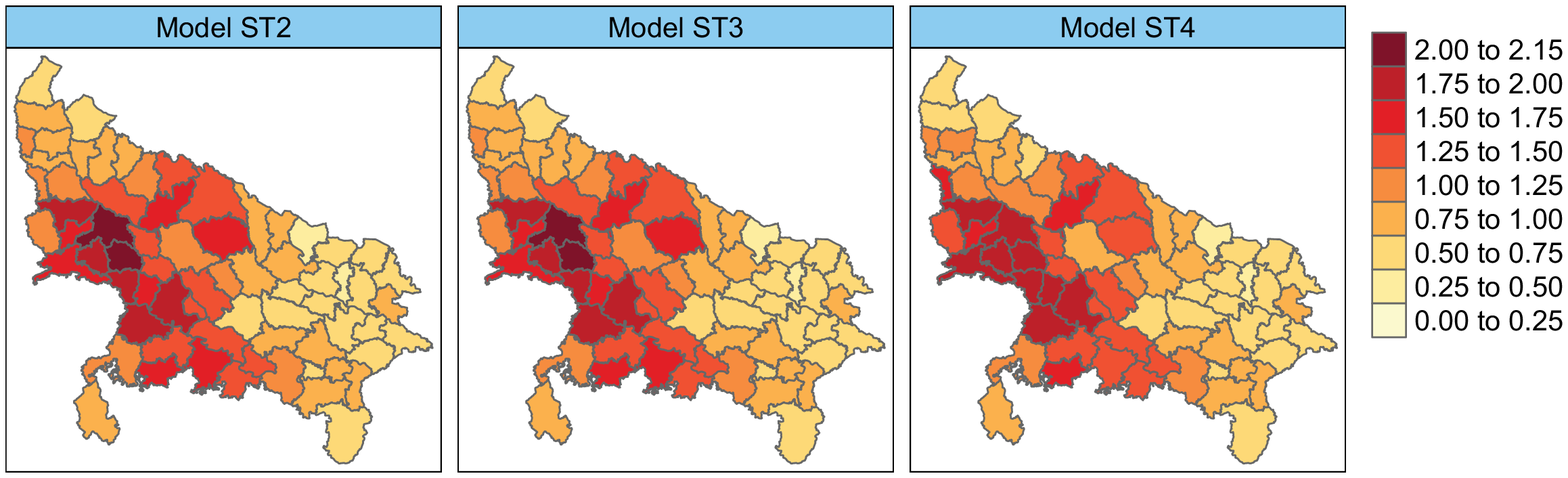}\\
\includegraphics[width=\textwidth]{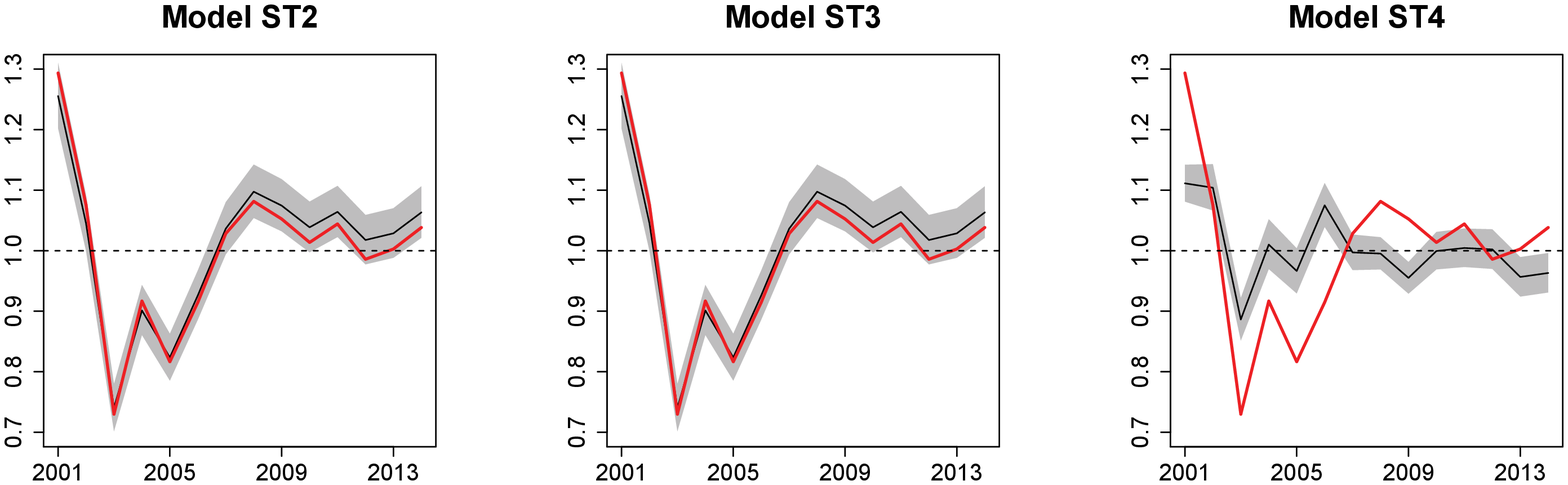}\\
\includegraphics[width=\textwidth]{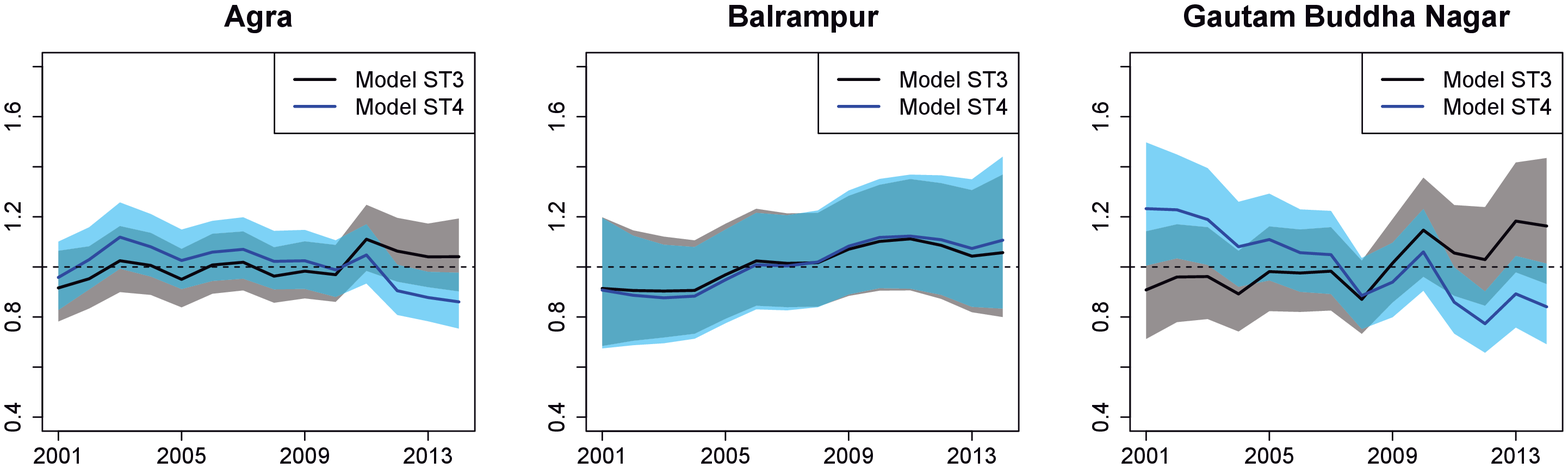}
\caption{Maps of posterior spatial patterns (top row) and posterior temporal patterns (middle row) obtained with models ST2, ST3, and ST4. Red lines (middle row) are the global standardized mortality ratios. Posterior spatio-temporal patterns (bottom row) obtained with Models ST3 and ST4 are shown for three districts (Agra, Balrampur and Gautam Buddha Nagar). Results are from the INLA fit.}
\label{figB3}
\end{figure}

Figure \ref{figB3} shows the posterior spatial patterns (top row), the posterior temporal patterns (middle row) obtained from Models ST2, ST3, and ST4 fitted with INLA (see \citealp{adin2017smoothing}), and posterior spatio-temporal patterns for three districts, Agra, Balrampur, and Gautam Buddha Nagar (bottom row).  While the posterior spatial patterns are quite similar for all models (top row), the posterior temporal and spatio-temporal patterns differ. The temporal patterns obtained with Models ST2 and ST3 are identical, while the temporal pattern obtained with Model ST4 is clearly different and does not track the global standardized mortality ratios (red line). Regarding posterior spatio-temporal patterns (space-time interactions), some areas present mild differences between Models ST3 and ST4 (e.g., Agra) and others exhibit negligible differences (Balrampur), but some districts show striking differences (Gautam Buddha Nagar).  In general, most districts have modest differences in the spatio-temporal component (not shown). %INLA relative risk estimates (posterior means) obtained with models ST3 and ST4 in the %same three districts shown in Figure \ref{figB3} can be found in in Figure \ref{figB4}.

%\newpage

\begin{figure}[!t]
\includegraphics[width=\textwidth]{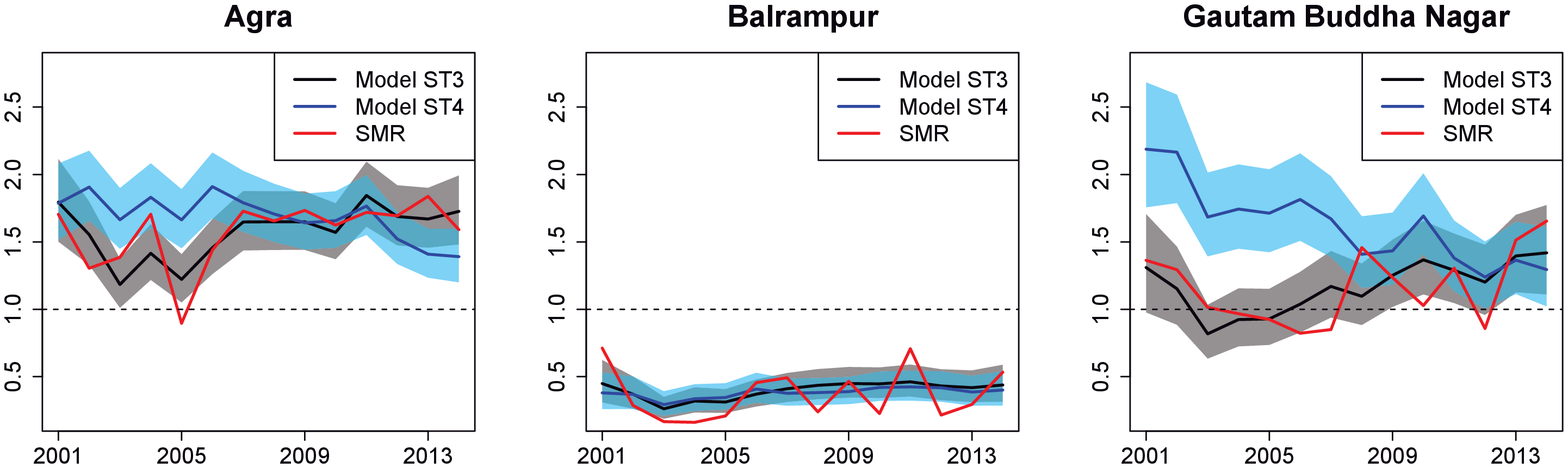}
\caption{Final risk estimates obtained with models ST3 and ST4 and INLA in three districts, Agra, Balrampur, and Gautam Buddha Nagar. Black lines and grey credible intervals correspond to Model ST3, blue lines and credible intervals to Model ST4. Red lines represent the crude standardized mortality ratios.}
\label{figB4}
\end{figure}

Figure \ref{figB4} displays the INLA relative risk estimates (posterior means) obtained with models ST3 and ST4 in the same three districts shown in Figure \ref{figB3}, Agra, Balrampur, and Gautam Buddha Nagar. Black lines and grey credible intervals are from Model ST3, while blue lines and blue credible intervals are from Model ST4. Standardized mortality ratios are shown in red. The differences in risks between Models ST3 and ST4 in Agra and Gautam Buddha Nagar are due to both the temporal and spatio-temporal components, while the differences in Balrampur are due to the temporal component. Given that the temporal pattern is common to all districts, it seems striking that the differences in risk in Balrampur are very small in comparison to Agra and Gautam Buddha Nagar. The reason is that the risk estimate is the product of the spatial, temporal, and spatio-temporal components. In Balrampur, the spatial component is small (between 0.25 and 0.50) whereas in Agra and Gautam Buddha Nagar the spatial relative risk is greater than one. Consequently, differences in risk are softened in Balarampur and accentuated in Agra and Gautam Buddha Nagar.

\end{document}